\documentclass[final,3p]{elsarticle}
 \usepackage[latin1]{inputenc}
 \usepackage{graphics}
 \usepackage{graphicx}
 \usepackage{epsfig}
\usepackage{amssymb}
 \usepackage{amsthm}
 \usepackage{lineno}
 \usepackage{amsmath}
   \numberwithin{equation}{section}
\usepackage{mathrsfs}

\newtheorem{thm}{Theorem}[section]

\newtheorem{lem}[thm]{Lemma}

\newtheorem{defn}[thm]{Definition}
\newtheorem{rem}[thm]{Remark}
\newtheorem{exam}[thm]{Example}

\biboptions{sort&compress,square}
\begin{document}
\begin{frontmatter}
\author[rvt1]{Jian Wang}
\ead{wangj@tute.edu.cn}
\author[rvt2]{Yong Wang\corref{cor2}}
\ead{wangy581@nenu.edu.cn}

\cortext[cor2]{Corresponding author.}
\address[rvt1]{School of Science, Tianjin University of Technology and Education, Tianjin, 300222, P.R.China}
\address[rvt2]{School of Mathematics and Statistics, Northeast Normal University,
Changchun, 130024, P.R.China}

\title{The spectral torsion for the Connes type operator}
\begin{abstract}
This paper aims to  provide an explicit computation of the spectral torsion associated with the Connes type operator
on even dimension compact manifolds.
And we also extend the spectral torsion for the Connes type operator to
compact manifolds  with boundary.
\end{abstract}
\end{frontmatter}
\section{Introduction}
\label{1}

Let $(M, g)$ be a compact Riemannian manifold, for Dirac operator
 with a fully antisymetric torsion $D_{T}=D-\frac{i}{8}T_{jkl}\gamma^{j}\gamma^{k}\gamma^{l}$,
the torsion functional defined by Dabrowski et al.\cite{DSZ} is the trilinear Clifford multiplication by functional of differential one-forms $u,v,w$
\begin{align}
\mathscr{T}(u,v,w)={\rm Wres}\big(\hat{u}\hat{v}\hat{w}D_{T}|D_{T}|^{-n}\big),
\end{align}
which recovers the torsion of the linear
connection for a canonical spectral triple over a closed spin manifold.
 Dirac operators with torsion are by now well-established analytical tools in the study of special geometric structures. Ackermann and
  Tolksdorf \cite{AT} proved a generalized version of the well-known Lichnerowicz formula for the square of the
most general Dirac operator with torsion  $D_{T}$ on an even-dimensional spin manifold associated to a metric connection with torsion.
In \cite{PS1}, Pf$\ddot{a}$ffle and Stephan considered orthogonal connections with arbitrary torsion on compact Riemannian manifolds,
and for the induced Dirac operators, twisted Dirac operators and Dirac operators of Chamseddine-Connes type they computed the spectral
action. In \cite{WW2},  we proved the Kastler-Kalau-Walze type theorem associated to Dirac operators with torsion on compact manifolds
 with boundary, and gave two kinds of operator-theoretic explanations of the gravitational action in the case of lower dimensional compact manifolds
 with flat boundary. Recently, Dabrowski et al.\cite{DSZ} introduced a trilinear functional of differential one-forms for a finitely
summable regular spectral triple with a noncommutative residue, and proposed a plain, purely spectral method to
determine the torsion as the density of the torsion functional and imposed
the torsion-free condition for regular finitely summable spectral triples.
The method of producing these spectral functionals in \cite{DSZ} is the noncommutative residue (or be written as
noncommutative residue) density of a suitable power of the Laplace type operator multiplied by a pair of other differential operators.

The noncommutative residue plays a prominent role in noncommutative geometry.
For one-dimensional manifolds, the noncommutative residue was discovered by Adler \cite{MA}
 in connection with geometric aspects of nonlinear partial differential equations.
 For arbitrary closed compact $n$-dimensional manifolds, the noncommutative reside was introduced by Wodzicki in \cite{Wo,Wo1} using the theory of zeta functions of elliptic pseudodifferential operators.
In \cite{Co1}, Connes used the noncommutative residue to derive a conformal 4-dimensional Polyakov action analogue.
Furthermore, Connes \cite{Co2} made a challenging observation that the noncommutative residue of the square of the inverse of the
Dirac operator was proportional to the Einstein-Hilbert action.
Let $s$ be the scalar curvature and Wres denote  the noncommutative residue,  the Kastler-Kalau-Walze type
theorem \cite{KW,Ka} gives an operator-theoretic explanation of the gravitational action and says that for a $4-$dimensional closed spin manifold,
 there exists a constant $c_0$, such that
 \begin{equation*}
{\rm  Wres}(D^{-2})=c_0\int_Ms{\rm dvol}_M.
\end{equation*}
And then, Ackermann proved that
the Wodzicki residue ${\rm  Wres}(D^{-2})$ in turn is essentially the second  coefficient
of the heat kernel expansion of $D^{2}$ in \cite{Ac}.
 Furthermore, Fedosov et al. defined the noncommutative residue on Boutet de Monvel's algebra and proved that it was a
unique continuous trace for the case of manifolds with boundary \cite{FGLS}.
In \cite{S}, Schrohe gave the relation between the Dixmier trace and the noncommutative residue for
manifolds with boundary.
Then Wang generalized Connes' results to the case of manifolds with boundary,
and proved  Kastler-Kalau-Walze type theorems for lower-dimensional manifolds
with boundary associated with Dirac operator and signature operator \cite{Wa1,Wa2,Wa3}.
And then, Sitarz and Zajac \cite{SZ} investigated the spectral action for scalar perturbations of Dirac operators.
Iochum and Levy \cite{IL}computed the heat kernel coefficients for Dirac operators with one-form perturbations.
Wang \cite{Wa4} considered the arbitrary perturbations of Dirac operators, and established the associated Kastler-
Kalau-Walze theorem. We gave two Lichnerowicz type formulas for Dirac operators
and signature operators twisted by a vector bundle with a non-unitary connection\cite{WW2}.
And in \cite{WWW}, we got two Lichnerowicz type formulas for the Dirac-Witten operators, and
gave the proof of Kastler-Kalau-Walze type theorems for the Dirac-Witten operators on 4-dimensional and 6-dimensional compact manifolds with
(resp. without) boundary.

Notably, in the preceding the calculation of the Wodzicki residue
research to recover the torsion of the linear
connection,  Dabrowski et al. \cite{DSZ} showed that the spectral definition of torsion can be readily extended to the noncommutative
case of spectral triples. By twisting the spectral triple of a Riemannian spin manifold, Martinetti et al. showed how
to generate an orthogonal and geodesic preserving torsion from a torsionless
Dirac operator in \cite{MNZ}. Above studied operators are the Connes type operators in this paper,
 motivated by the spectral torsion  \cite{DSZ}, the noncommutative residue and
Kastler-Kalau-Walze type theorem,  we give an explicit computation of the spectral torsion associated with Connes type operator
on compact manifolds with (or without) boundary.
   We end the introduction with a brief discussion on the organization of the paper.
    In Sections 2-3, we give preliminaries for the Connes type operator
and explicit computation of the spectral torsion associated with the Connes type operator.
  In Sections 4, for several different examples of Connes type operators with
  the trilinear Clifford multiplication by functional of differential one-forms, we compute
 the spectral torsion for manifolds with (or without) boundary..

\section{Preliminaries for the Connes type operator}

The purpose of this section is to specify the Connes type operator.
 Let $M$ be a compact oriented Riemannian manifold of even dimension $n=2m$, and
 let  $\mathcal{E}$ be the spinor bundle over $M$ wiht the Hermitian structure and the spin connection $\nabla^{\mathcal{E}}$.
\begin{defn}
A Dirac operator $D$ on a $\mathbb{Z}_{2}$-graded vector bundle $\mathcal{E}$ is a
first-order differential operator of odd parity on $\mathcal{E}$,
 \begin{align}
D:\Gamma(M,\mathcal{E}^{\pm})\rightarrow\Gamma(M,\mathcal{E}^{\mp}),
\end{align}
such that $D^{2}$ is a generalized Laplacian.
\end{defn}
 Let $\nabla^L$ denote the Levi-Civita connection about $g^M$. In the
 fixed orthonormal frame $\{e_1,\cdots,e_n\}$ and natural frames $\{\partial_1,\cdots,\partial_n\}$ of $TM$,
the connection matrix $(\omega_{s,t})$ defined by
\begin{equation}
\nabla^L(e_1,\cdots,e_n)= (e_1,\cdots,e_n)(\omega_{s,t}).
\end{equation}
 Then the Dirac operator has the form
\begin{equation}
D=\sum^n_{i=1}c(e_i)\nabla_{e_i}^{S}=\sum^n_{i=1}c(e_i)\Big[e_i
-\frac{1}{4}\sum_{s,t}\omega_{s,t}(e_i)c(e_s)c(e_t)\Big].
\end{equation}
The Levi-Civita connection
$\nabla^L: \Gamma(TM)\rightarrow \Gamma(T^{*}M\otimes TM)$ on $M$ induces a connection
$\nabla^{S}: \Gamma(S)\rightarrow \Gamma(T^{*}M\otimes S).$ By adding a additional torsion term $t\in\Omega^{1}(M,{\rm{ End}}(TM))$ we
obtain a new covariant derivative
 \begin{equation}
\widetilde{\nabla}:=\nabla^L +t
\end{equation}
on the tangent bundle $TM$. Since $t$ is really a one-form on $M$ with values in the bundle of skew endomorphism $Sk(TM)$ in \cite{GHV},
$\widetilde{\nabla}$  is in fact compatible with the Riemannian metric $g$ and therefore also induces a connection $\widetilde{\nabla}^{S}:=\nabla^{S}+T$
on the spinor bundle. Here $T\in\Omega^{1}(M, {\rm{ End}} S)$ denotes the `lifted' torsion term $t\in\Omega^{1}(M, {\rm{ End}}(TM))$.
 As Clifford multiplication by any 3-form $T$ is self-adjoint we have
 \begin{equation}
D_{T}\psi=D\psi+\frac{3}{2}T\cdot\psi-\frac{n-1}{2}Y\cdot\psi,
\end{equation}
where $c(T)=\sum _{1\leq\alpha<\beta<\gamma\leq n}
 T_{\alpha \beta \gamma}c(e_{\alpha})c(e_{\beta})c(e_{\gamma})$ and the Clifford multiplication by the vector field $Y$ is skew-adjoint,
the hermitian product on the spinor bundle one observes
that $D_{T}$ is symmetric with respect to the natural $L^{2}$-scalar product on spinors if and only
if the vectorial component of the torsion vanishes, $Y\equiv 0$.

Denote by $\mathbb{C}l(M)$ the set of smooth sections of the vector bundle
with fibre over  $x\in M$ the Clifford algebra over the cotangent space of $x$ (a $\mathbb{Z}_{2}$ graded complex algebra
 and $C^{\infty}(M)$-module), we now consider an additional smooth
vector bundle $F$ over $M$ (with $C^{\infty}(M)$-module of smooth sections $W$),  equipped with a connection $\nabla^{F}$,
with corresponding curvature-tensor $R^{F}$. We consider the tensor product vector bundle $S(TM)\otimes F$
 equiped with the compound connection:
 \begin{equation}
\nabla^{ S(TM)\otimes F}= \nabla^{ S(TM)}\otimes \rm{id}_{ F}+ \rm{id}_{ S(TM)}\otimes \nabla^{F}.
\end{equation}
The corresponding twisted Dirac operator $D_{F}$ is locally specified as follows:
 \begin{equation}
D_{F}=\sum_{i}^{n}c(e_{i})\nabla^{ S(TM)\otimes F}_{e_{i}},
\end{equation}
where $c(e_{i})$ denotes the Clifford action of the dual
of the vector field on $\mathcal{E}$, which satisfies the relation
$
 c(e_{i})c(e_{j})+c(e_{j})c(e_{i})=-2\langle e_{i}, e_{j}  \rangle.
$
 Then we have the following definition.
\begin{defn}
The Connes type operator  is the first order differential operator on $ \Gamma(M,S(TM)\otimes F)$  given by the formula
 \begin{align}
\widetilde{D}_{F}= \sum_{i}^{n}c(e_{i})\nabla^{ S(TM)\otimes F}_{e_{i}}+\Box\otimes \Phi ,
\end{align}
where $\Box$ denotes the Clifford multiplication by any form,  $\Phi \in  \Gamma(M, {\rm{ End}} (F))$ and $\Phi=\Phi^{*}$.
\end{defn}
It is natural to study the Connes type operator with these extremal properties.
\begin{exam}
Suppose that  $\Box=\gamma$, the grading operator on $S(TM)$, then
 $D_{F}+\gamma \otimes \Phi$  is the Connes type operator.
\end{exam}
\begin{exam} Let $\Box=c(X)\gamma$ be the Clifford multiplication by the vector field $X$,
then $D_{F}+c(X)\gamma \otimes \Phi$  is the
twisted  type operator  in \cite{MNZ}.
\end{exam}
\begin{exam} Let $\Box=\sqrt{-1}c(T)\gamma$ be the Clifford multiplication by the 3-form $T$ and the grading operator $\gamma$ on $S(TM)$,
then $D_{F}+\sqrt{-1}c(T)\gamma \otimes \Phi$  is the twisted  type operator with torsion.
\end{exam}
\begin{exam}
Assume  $\Box=(c(T)+\sqrt{-1}c(Y))$ with  $c(T)=\sum _{1\leq\alpha<\beta<\gamma\leq n}
 T_{\alpha \beta \gamma}c(e_{\alpha})c(e_{\beta})c(e_{\gamma})$ and the Clifford multiplication by the vector field $Y$ is skew-adjoint, then
$D_{F}+(c(T)+\sqrt{-1}c(Y)) \otimes \Phi$  is the general twisted Dirac operator with torsion.
\end{exam}

\section{Trilinear functional  for the  Connes type operator with torsion}

In this section we want to consider the trilinear functional  for the  Connes type operator with torsion.
To avoid technique terminology we only state our results for
compact oriented Riemannian manifold of even dimension $n=2m$ by using the trace of  Connes type operator and the noncommutative residue density.
\begin{defn}\cite{DSZ}
For $D_{F}$ given by (2.7), the trilinear Clifford multiplication by functional of differential one-forms
$c(u),c(v),c(w)$
\begin{align}
\mathscr{T}(c(u),c(v),c(w))={\rm Wres}\big(c(u)c(v)c(w)(D_{F}+\Box\otimes \Phi )^{-n+1}\big)
\end{align}
is called torsion functional associated with the Connes type operator.
\end{defn}
We divide the computation of the spectral torsion associated with Connes operator into two steps, according to the dimension $n$.

{\bf  Step-1}. Explicit representation of the symbols of $D_{F}^{2}$ and $\widetilde{D}_{F}^2$.

Let  $\{e_{i}\}(1\leq i,j\leq n)$ $(\{\partial_{i}\})$ be the orthonormal
frames (natural frames respectively ) on  $TM$, then
 \begin{equation}
D_{F}=\sum_{i,j}g^{ij}c(\partial_{i})\nabla^{ S(TM)\otimes F}_{\partial_{j}}=\sum_{j}^{n}c(e_{j})\nabla^{ S(TM)\otimes F}_{e_{j}},
\end{equation}
where $\nabla^{S(TM)\otimes F}_{\partial_{j}}=\partial_{j}+\sigma_{j}^{s}+\sigma_{j}^{F}$ and
$\sigma_{j}^{s}=\frac{1}{4}\sum_{j,k}\langle \nabla^{L}_{\partial_{i}}e_{j}, e_{k}\rangle c(e_{j})c(e_{k})$,
$\sigma_{j}^{F}$ is the connection  matrix of $\nabla^{F}$.

In what follows, using the Einstein sum convention for repeated index summation:
$\partial^j=g^{ij}\partial_i$,$\sigma^i=g^{ij}\sigma_{j}$, $\Gamma_{k}=g^{ij}\Gamma_{ij}^{k}$,
 and we omit the summation sign, from (6a) in \cite{Ka}, we have
\begin{align}
D_{F}^{2}=&-g^{ij}\partial_{i}\partial_{j}-2\sigma^{j}_{S(TM)\otimes F}\partial_{j}+\Gamma^{k}\partial_{k}
                        -g^{ij}\Big[\partial_{i}(\sigma^{j}_{S(TM)\otimes F}) +\sigma^{i}_{S(TM)\otimes F}\sigma^{j}_{S(TM)\otimes F}
                  \nonumber\\
                  & -\Gamma_{ij}^{k}\sigma_{S(TM)\otimes F}^{k}\Big]+\frac{1}{4}s+\frac{1}{2}\sum_{i\neq j} R^{F}(e_{i},e_{j})c(e_{i})c(e_{j}),
\end{align}
where $s$ is the scaler curvature.
By (3.3) and direct calculation, we obtain
 \begin{align}
\widetilde{D}_{F}^2=&\Big(D_{F}+\Box\otimes \Phi \Big)^{2}\nonumber\\
=&D_{F}^{2}+D_{F}(\Box\otimes \Phi)+\Box\otimes \Phi D_{F}+\big(\Box\otimes \Phi\big)^{2}\nonumber\\
=&-g^{ij}\partial_{i}\partial_{j}+\Gamma^{k}\partial_{k}
   +g^{ij}\big[-2\sigma_{i}^{ S(TM)\otimes F}
   +c(\partial_{i})\big(\Box \otimes \Phi\big) +\big(\Box\otimes \Phi \big)c(\partial_{i})   \big]\partial_{j}\nonumber\\
    &+g^{ij}\big[-\partial_{i}(\sigma_{j}^{ S(TM)\otimes F} )-\sigma_{i}^{ S(TM)\otimes F}\sigma_{j}^{ S(TM)\otimes F}
      + \Box\otimes \Phi c(\partial_{i})\sigma_{j}^{ S(TM)\otimes F}\nonumber\\
&+ \Gamma^{k}_{ij}\sigma_{k}^{ S(TM)\otimes F}
+c(\partial_{i})\partial_{j}(\Box \otimes \Phi)
    +c(\partial_{i})\sigma_{j}^{ S(TM)\otimes F}(\Box \otimes \Phi)\big]\nonumber\\
&+\Box^{2}\otimes \Phi^{2}+\frac{1}{4}s+\frac{1}{2}\sum_{i\neq j}R^{F}(e_{i},e_{j})c(e_{i})c(e_{j}).
\end{align}

{\bf  Step-2}. Explicit representation  of $\sigma_{-n}(c(u)c(v)c(w)\widetilde{D}_{F}^{-n+1})
=\sigma_{-n}\big(c(u)c(v)c(w)(D_{F}+\Box\otimes \Phi )^{-n+1}\big)$.

Let the cotangent vector $\xi=\sum \xi_jdx_j$ and
$\xi^j=g^{ij}\xi_i$.
Decompose the operators $\widetilde{D}_{F}^2$ and $\widetilde{D}_{F}$ by different orders as
\begin{align}
\sigma(\widetilde{D}_{F}^{2})=\sigma_{2}(\widetilde{D}_{F}^{2})
+\sigma_{1}(\widetilde{D}_{F}^{2})+\sigma_{0}(\widetilde{D}_{F}^{2});~~
 \sigma(\widetilde{D}_{F}^{1})=\sigma_{1}(\widetilde{D}_{F}^{1})+\sigma_{0}(\widetilde{D}_{F}^{1}).
\end{align}
Note that the arguments in the proof of Lemma 3.1 in \cite{WW1} can be used to compute
the symbolic representation for the Connes type operator without any change, where the same conclusion was derived for Dirac operator.
By (3.4), (3.5) and Lemma 3.1 \cite{WW1}, we get
 \begin{align}
\sigma_{-2}(\widetilde{D}_{F}^{-2})=|\xi|^{-2},~~\sigma_{-(2m-2)}(\widetilde{D}_{F}^{-(2m-2)})=(|\xi|^2)^{1-m},
~~\sigma_{-2m+1}(\widetilde{D}_{F}^{-2m+1})=\frac{\sqrt{-1}c(\xi)}{|\xi|^{2m}}.
\end{align}
Then by (3.8) in \cite{WW1}, we have
 \begin{align}
 \sigma_{-1-2m}(\widetilde{D}_{F}^{-2m})=&m\sigma_2(\widetilde{D}_{F}^2)^{(-m+1)}\sigma_{-3}(\widetilde{D}_{F}^{-2}) \nonumber\\
 &-\sqrt{-1} \sum_{k=0}^{m-2}\sum_{\mu=1}^{2m+2}
\partial_{\xi_{\mu}}\sigma_{2}^{-m+k+1}(\widetilde{D}_{F}^2)
\partial_{x_{\mu}}\sigma_{2}^{-1}(\widetilde{D}_{F}^2)(\sigma_2(\widetilde{D}_{F}^2))^{-k}.
\end{align}
In the same way we obtain
\begin{align}
 \sigma_{-2m}(\widetilde{D}_{F}^{1-2m})=&\sigma_{-2m}(\widetilde{D}_{F}^{-2m}\cdot \widetilde{D}_{F})
=\Big\{\sum_{|\alpha|=0}^{+\infty}(-\sqrt{-1})^{|\alpha|}
\frac{1}{\alpha!}\partial^\alpha_\xi[\sigma(\widetilde{D}_{F}^{-2m})]\partial^\alpha_x[\sigma(\widetilde{D}_{F})]\Big\}_{-2m}   \nonumber\\
 =&\sigma_{-2m}(\widetilde{D}_{F}^{-2m})\sigma_0(\widetilde{D}_{F})+\sigma_{-2m-1}(\widetilde{D}_{F}^{-2m})\sigma_1(\widetilde{D}_{F})\nonumber\\
&+\sum_{|\alpha|=1}(-\sqrt{-1})
\partial^\alpha_\xi[\sigma_{-2m}(\widetilde{D}_{F}^{-2m})]\partial^\alpha_x[\sigma_1(\widetilde{D}_{F})]   \nonumber\\
 =&|\xi|^{-2m}\sigma_0(\widetilde{D}_{F})+\sum_{j=1}^{2m+2}\partial_{\xi_j}(|\xi|^{-2m})\partial_{x_j}c(\xi)+
\Big[m\sigma_2(\widetilde{D}_{F}^2)^{-m+1}\sigma_{-3}(\widetilde{D}_{F}^{-2})  \nonumber\\
&-\sqrt{-1} \sum_{k=0}^{m-2}\sum_{\mu=1}^{2m+2}
\partial_{\xi_{\mu}}\sigma_{2}^{-m+k+1}(\widetilde{D}_{F}^2)
\partial_{x_{\mu}}\sigma_{2}^{-1}(\widetilde{D}_{F}^2)(\sigma_2(\widetilde{D}_{F}^2))^{-k}\Big ]\sqrt{-1}c(\xi).
\end{align}
Also, straightforward computations yield
\begin{align}
 &\sigma_{-2m}(c(u)c(v)c(w)\widetilde{D}_{F}^{1-2m})\nonumber\\
=&c(u)c(v)c(w)\Big\{|\xi|^{-2m}\sigma_0(\widetilde{D}_{F})+\sum_{j=1}^{2m+2}\partial_{\xi_j}(|\xi|^{-2m})\partial_{x_j}c(\xi)+
\Big[m\sigma_2(\widetilde{D}_{F}^2)^{-m+1}\sigma_{-3}(\widetilde{D}_{F}^{-2})  \nonumber\\
&-\sqrt{-1} \sum_{k=0}^{m-2}\sum_{\mu=1}^{2m+2}
\partial_{\xi_{\mu}}\sigma_{2}^{-m+k+1}(\widetilde{D}_{F}^2)
\partial_{x_{\mu}}\sigma_{2}^{-1}(\widetilde{D}_{F}^2)(\sigma_2(\widetilde{D}_{F}^2))^{-k}\Big ]\sqrt{-1}c(\xi)\Big\}.
\end{align}
Substituting (3.9) into (3.1) yields the torsion functional for the Connes type operator with the trilinear
Clifford multiplication by functional of differential one-forms
$c(u),c(v),c(w)$ as follows:
\begin{align}
\mathscr{T}(c(u),c(v),c(w))=\int_M\int_{|\xi|=1}{\rm Tr}_{S(TM)\otimes F}
[\sigma_{-2m}\big(c(u)c(v)c(w)(D_{F}+\Box\otimes \Phi )^{-2m+1}\big)]\sigma(\xi){\rm d}x.
\end{align}

\section{The computation of the spectral torsion for the Connes type operator}

In this section, we develop the several examples of generalized torsion functional for Connes operators with the trilinear Clifford multiplication
by functional of differential one-forms. For these perturbed Connes operators with torsion, we obtain the torsion functional representation
 by calculating the traces of such operators and the noncommutative residue density.

\subsection{Spectral torsion for the Connes type operator with 3-form}
We will assume in this section that $\Box\otimes \Phi=(c(T)+\sqrt{-1}c(Y))\otimes \Phi $, where $c(T)=\sum _{1\leq\alpha<\beta<\gamma\leq n}
 T_{\alpha \beta \gamma}c(e_{\alpha})\times c(e_{\beta})c(e_{\gamma})$ and the Clifford multiplication by the vector field $Y$ is skew-adjoint,
 and $\Phi \in  \Gamma(M, {\rm{ End}} (F))$.
Substituting  $\Box=c(T)+\sqrt{-1}c(Y) $ into (3.4) yields
 \begin{align}
&\Big(D_{F}+(c(T)+\sqrt{-1}c(Y))\otimes \Phi \Big)^{2}\nonumber\\
=&D_{F}^{2}+D_{F}(c(T)+\sqrt{-1}c(Y))\otimes \Phi+(c(T)+\sqrt{-1}c(Y))\otimes \Phi D_{F}+\big((c(T)+\sqrt{-1}c(Y))\otimes \Phi\big)^{2}\nonumber\\
=&-g^{ij}\partial_{i}\partial_{j}
   +\Gamma^{k}\partial_{k}+g^{ij}\big[-2\sigma_{i}^{ S(TM)\otimes F}+c(\partial_{i})\big((c(T)+\sqrt{-1}c(Y)) \otimes \Phi\big)\nonumber\\
   & +\big((c(T)+\sqrt{-1}c(Y))\otimes \Phi \big)c(\partial_{i})
     \big]\partial_{j}\nonumber\\
    &+g^{ij}\big[-\partial_{i}(\sigma_{j}^{ S(TM)\otimes F} )-\sigma_{i}^{ S(TM)\otimes F}\sigma_{j}^{ S(TM)\otimes F}
      + (c(T)+\sqrt{-1}c(Y))\otimes \Phi c(\partial_{i})\sigma_{j}^{ S(TM)\otimes F}\nonumber\\
&+ \Gamma^{k}_{ij}\sigma_{k}^{ S(TM)\otimes F}
+c(\partial_{i})\partial_{j}((c(T)+\sqrt{-1}c(Y)) \otimes \Phi)
    +c(\partial_{i})\sigma_{j}^{ S(TM)\otimes F}((c(T)+\sqrt{-1}c(Y)) \otimes \Phi)\big]\nonumber\\
&+(c(T)+\sqrt{-1}c(Y))^{2}\otimes \Phi^{2}+\frac{1}{4}s+\frac{1}{2}\sum_{i\neq j}R^{F}(e_{i},e_{j})c(e_{i})c(e_{j}).
\end{align}
The arguments from hierarchical approach
seems very useful to understand the main symbol for the Connes type operator.
We first sketch the $-n$ order symbolic representation of $c(u)c(v)c(w)(D_{F}+(c(T)+\sqrt{-1}c(Y)) \otimes \Phi$.
Then we get the  spectral torsion for Connes operator   $D_{F}+(c(T)+\sqrt{-1}c(Y)) \otimes \Phi $
 by calculating the traces of such operators.
The key point in our case is to use  the method of Section 3 to get the noncommutative residue density
of a suitable power of the Laplace type operator multiplied by 3-form.

{\bf  (1)}. The representation of $\sigma_{-n}\big(c(u)c(v)c(w)(D_{F}+(c(T)+\sqrt{-1}c(Y)) \otimes \Phi )^{-n+1}\big)$.

Write the Dirac operators $D^2$ and $D^{-1}$ by different orders as
  \begin{equation}
D_x^{\alpha}=(-\sqrt{-1})^{|\alpha|}\partial_x^{\alpha};~\sigma(D^{2})=p_2+p_1+p_0;
~\sigma(D^{-1})=\sum^{\infty}_{j=1}q_{-j}.
\end{equation}
By the composition formula of psudodifferential operators, we obtain
\begin{align}
1=&\sigma(D^{2}\circ D^{-2})=\sum_{\alpha}\frac{1}{\alpha!}\partial^{\alpha}_{\xi}[\sigma(D^{2})]D^{\alpha}_{x}[\sigma(D^{-2})]\nonumber\\
=&(p_2+p_1+p_0)(q_{-2}+q_{-3}+q_{-4}+\cdots)\nonumber\\
&+\sum_j(\partial_{\xi_j}p_2+\partial_{\xi_j}p_1+\partial_{\xi_j}p_0)(
D_{x_j}q_{-2}+D_{x_j}q_{-3}+D_{x_j}q_{-4}+\cdots)\nonumber\\
&+\sum_{i,j}(\partial_{\xi_i}\partial_{\xi_j}p_2+\partial_{\xi_i}\partial_{\xi_j}p_1+\partial_{\xi_i}\partial_{\xi_j}p_0)(
D_{x_i}D_{x_j}q_{-2}+D_{x_i}D_{x_j}q_{-3}+D_{x_i}D_{x_j}q_{-4}+\cdots)\nonumber\\
=&p_2q_{-2}+(p_1q_{-2}+p_2q_{-3}+\sum_j\partial_{\xi_j}p_2D_{x_j}q_{-2})\nonumber\\
   &+(p_0q_{-2}+p_1q_{-3}+p_2q_{-4}+\sum_j\partial_{\xi_j}p_2D_{x_j}q_{-3}+\sum_{i,j}\partial_{\xi_i}\partial_{\xi_j}p_2
  D_{x_i}D_{x_j}q_{-2})   +\cdots .
\end{align}
Then comparing the same order terms on both side of (4.3) gives
\begin{align}
q_{-2}=&p_1^{-2};  \\
q_{-3}=&-p_2^{-1}\Big[p_1q_{-2}+\sum_j\partial_{\xi_j}p_2D_{x_j}q_{-2}\Big].
\end{align}
From (4.1),(4.4), (4.5) and Lemma 1 in \cite{Wa3}, we obtain
\begin{align}
&\sigma_{-3}\Big((D_{F}+(c(T)+\sqrt{-1}c(Y))\otimes \Phi)^{-2}\Big)\nonumber\\
=&-\sqrt{-1}|\xi|^{-4}\xi_k(\Gamma^k-2\delta^k)-\sqrt{-1}|\xi|^{-6}2\xi^j\xi_\alpha\xi_\beta
    \partial_jg^{\alpha\beta}-|\xi|^{-4}\Gamma^{k}\partial_{k}\nonumber\\
    &-|\xi|^{-4}g^{ij}\big[-2\sigma_{i}^{ S(TM)\otimes F}
   +c(\partial_{i})\big((c(T)+\sqrt{-1}c(Y)) \otimes \Phi\big)\nonumber\\
    &+\big((c(T)+\sqrt{-1}c(Y))\otimes \Phi \big)c(\partial_{i})   \big]\sqrt{-1}\xi_{j}.
\end{align}

For any fixed point $x_0\in M$, we can choose the normal coordinates
$U$ of $x_0$ in $M$ and compute $\sigma_{-n}\big(c(u)c(v)c(w)(D_{F}+(c(T)+\sqrt{-1}c(Y)) \otimes \Phi )^{-n+1}\big)$.
Let $\{E_1,\cdots,E_n\}$ be the canonical basis of $\mathbb{R}^n$ and $c(E_i)\in {\rm cl}_{\bf C}(n)\cong {\rm Hom}(\wedge^*
_{\bf C}(\frac{n}{2}),\wedge^* _{\bf C}(\frac{n}{2}))$ be the
Clifford action. Then
\begin{align}
c(e_i)=[(\sigma,c(E_i))];~ c( e_i)[(\sigma,f_i)]=[(\sigma,c(E_i)f_i)];~
\frac{\partial}{\partial x_i}=[(\sigma,\frac{\partial}{\partial
x_i})],
\end{align}
then we have $ \partial_{x_j}(c( e_i))=0$ in the above frame. In terms of normal coordinates about $x_{0}$ one has:
$\sigma^{j}_{S(TM)}(x_{0})=0$ $e_{j}\big(c(e_{i})\big)(x_{0})=0$, $\Gamma^{k}(x_{0})=0$.
 \begin{lem} In the normal coordinates $U$ of $x_0$ in $M$,
 \begin{align}
 \partial_{x_j}(\sigma_{-2}(\widetilde{D}_{F}^{-2}))(x_0)=0,\partial_{x_j}(c(\xi))(x_0)=0.
\end{align}
\end{lem}
 \begin{proof}
By (3,6), we get
 \begin{align}
 \partial_{x_j}(\sigma_{-2}(\widetilde{D}_{F}^{-2}))(x_0)= \partial_{x_j}(|\xi|^{-2})(x_0)= -2|\xi|^{-4} \partial_{x_j}(|\xi|^{2})(x_0)=0.
\end{align}
Also, straightforward computations yield
 \begin{align}
\partial_{x_j}(c(\xi))(x_0)=&\sum_{l=1}^{n}\partial_{x_j}(\xi_{l}c(dx^{l}))(x_0)
=\sum_{l=1}^{n}\xi_{l}\partial_{x_l}\big(c((dx^{l},e_{\alpha})e_{\alpha})\big)(x_0)\nonumber\\
=&\sum_{l=1}^{n}\xi_{l}\Big(\partial_{x_j}\big((dx^{l},e_{\alpha})\big)c(e_{\alpha})
+\sum_{l=1}^{n}\xi_{l}(dx^{l},e_{\alpha})\partial_{x_j}\big(c(e_{\alpha})\big)\Big)(x_0)=0.
\end{align}
\end{proof}
 \begin{lem} For the connection $\nabla^{F}$ on $F$, we choose a local frame of $F$ around $x_0$ in $M$,
 such that $ \sigma^{j}_{F}(x_{0})=0$.
\end{lem}
Substituting above results into the formula (3.9) gives the desired result
\begin{align}
 &\sigma_{-2m}(c(u)c(v)c(w)\widetilde{D}_{F}^{1-2m})(x_{0})|_{|\xi|=1} \nonumber\\
=&c(u)c(v)c(w)(c(T)+\sqrt{-1}c(Y)) \otimes \Phi
+mg^{ij}c(u)c(v)c(w)\Big[c(\partial_{i})\big((c(T)+\sqrt{-1}c(Y)) \otimes \Phi\big) \nonumber\\
&+\big((c(T)+\sqrt{-1}c(Y))\otimes \Phi \big)c(\partial_{i})   \Big]\xi_{j}c(\xi).
\end{align}

{\bf  (2)}. Spectral torsion for the Connes type operator   $D_{F}+(c(T)+\sqrt{-1}c(Y)) \otimes \Phi $

Let $u=\sum_{i=1}^{n}u_{i}e_{i}$, $v=\sum_{j=1}^{n}v_{j}e_{j} $, $w=\sum_{l=1}^{n}w_{l}e_{l} $,
where $\{e_{1},e_{2},\cdots,e_{n}\}$ is the orthogonal basis about $g^{TM}$,
then  $c(u)= \sum_{i=1}^{n}u_{i}c(e_{i})$,$c(v)= \sum_{j=1}^{n}v_{j}c(e_{j})$, $c(w)= \sum_{l=1}^{n}w_{l}c(e_{l})$.
We collect here, for the reader¡¯s convenience, all necessary facts on the calculation of noncommutative residue density.
 \begin{lem}
The following identities hold:
 \begin{align}
 &{\rm{Tr}}\big(c(u)c(v)c(w)c(Y)\big)=[g(v,w)g(u,Y)-g(u,w)g(v,Y)+g(u,v)g(w,Y)]{\rm{Tr}}(\rm{Id}) ;\\
 &{\rm{Tr}}\big(c(u)c(v)c(w)c(T)\big)=T(u,v,w){\rm{Tr}}(\rm{Id}).
\end{align}
\end{lem}
\begin{proof}
By the relation of the Clifford action and $ {\rm{Tr}}(AB)= {\rm{Tr}}(BA) $, we have the equality:
 \begin{align}
 {\rm{Tr}}\big(c(u)c(v)c(w)c(Y)\big)=&\sum_{i,j,k,l=1}^{n}u_{i}v_{j}w_{k}Y_{l}{\rm{Tr}}\big(c(e_{i})c(e_{j})c(e_{k})c(e_{l})\big)\nonumber\\
           =&\sum_{i,j,k,l=1}^{n}u_{i}v_{j}w_{k}Y_{l}(-\delta_{i}^{k}\delta_{j}^{l}+\delta_{i}^{l}\delta_{j}^{k}
           +\delta_{i}^{j}\delta_{k}^{l}){\rm{Tr}}(\rm{Id})\nonumber\\
           =&\sum_{i,j=1}^{n}(-u_{i}v_{j}w_{i}Y_{j} +u_{i}v_{j}w_{j}Y_{i}+u_{i}v_{i}w_{k}Y_{k}){\rm{Tr}}(\rm{Id})\nonumber\\
  =&[g(v,w)g(u,Y)-g(u,w)g(v,Y)+g(u,v)g(w,Y)]{\rm{Tr}}(\rm{Id}) .
\end{align}
Applying   the Clifford action  property  again, we get a cyclic representation of the trace relation
 \begin{align}
 &{\rm{Tr}}\big(c(u)c(v)c(w)c(T)\big)\nonumber\\
 =&\sum_{1 \leq  j,k,l\leq n}\sum_{\alpha,\beta,\gamma=1}^{n}T(e_{j},e_{k},e_{l})u_{\alpha}v_{\beta}w_{\gamma}
{\rm{Tr}}\big(c(e_{\alpha})c(e_{\beta})c(e_{\gamma})c(e_{j})c(e_{k})c(e_{l})\big)\nonumber\\
   =&\sum_{1 \leq  j,k,l\leq n}\sum_{\alpha,\beta,\gamma=1}^{n}T(e_{j},e_{k},e_{l})u_{\alpha}v_{\beta}w_{\gamma}
{\rm{Tr}}\big(c(e_{\alpha})c(e_{\beta})c(e_{\gamma})c(e_{l})c(e_{j})c(e_{k})\big)\nonumber\\
  =&\sum_{1 \leq  j,k,l\leq n}\sum_{\alpha,\beta,\gamma=1}^{n}T(e_{j},e_{k},e_{l})u_{\alpha}v_{\beta}w_{\gamma}
{\rm{Tr}}\Big(c(e_{\alpha})c(e_{\beta})\big(-c(e_{l})c(e_{\gamma})-2\delta_{\gamma}^{l}\big)c(e_{j})c(e_{k})\Big)\nonumber\\
=&\cdots   \nonumber\\
=&-\sum_{1 \leq  j,k,l\leq n}\sum_{\alpha,\beta,\gamma=1}^{n}T(e_{j},e_{k},e_{l})u_{\alpha}v_{\beta}w_{\gamma}
{\rm{Tr}}\big(c(e_{\alpha})c(e_{\beta})c(e_{\gamma})c(e_{j})c(e_{k})c(e_{l})\big)\nonumber\\
&+2\sum_{1 \leq  j,k,l\leq n}\sum_{\alpha,\beta,\gamma=1}^{n}T(e_{j},e_{k},e_{l})u_{\alpha}v_{\beta}w_{\gamma}
\Big(\delta_{\alpha}^{l}\delta_{\beta}^{j }\delta_{\gamma}^{k}-\delta_{\alpha}^{l}\delta_{\beta}^{k}\delta_{\gamma}^{j}
+\delta_{\alpha}^{l}\delta_{\beta}^{\gamma}\delta_{j}^{k}\nonumber\\
&+\delta_{\gamma}^{l}\delta_{\alpha}^{j}\delta_{\beta}^{k}-\delta_{\gamma}^{l}\delta_{\alpha}^{k}\delta_{\beta}^{j}
+\delta_{\gamma}^{l}\delta_{\alpha}^{\beta}\delta_{j}^{k}
-\delta_{\beta}^{l}\delta_{\alpha}^{j}\delta_{\gamma}^{k}+\delta_{\beta}^{l}\delta_{\alpha}^{k}\delta_{\gamma}^{j}
+\delta_{\beta}^{l}\delta_{\alpha}^{\gamma}\delta_{j}^{k} \Big){\rm{Tr}}(\rm{Id}),
\end{align}
By the generating relation, we see that the left hand side of (4.15) equals
 \begin{align}
 &{\rm{Tr}}\big(c(u)c(v)c(w)c(T)\big)\nonumber\\
 =&\sum_{1 \leq  j,k,l\leq n}\sum_{\alpha,\beta,\gamma=1}^{n}T(e_{j},e_{k},e_{l})u_{\alpha}v_{\beta}w_{\gamma}
\Big(\delta_{\alpha}^{l}\delta_{\beta}^{j }\delta_{\gamma}^{k}-\delta_{\alpha}^{l}\delta_{\beta}^{k}\delta_{\gamma}^{j}
+\delta_{\alpha}^{l}\delta_{\beta}^{\gamma}\delta_{j}^{k}\nonumber\\
&+\delta_{\gamma}^{l}\delta_{\alpha}^{j}\delta_{\beta}^{k}-\delta_{\gamma}^{l}\delta_{\alpha}^{k}\delta_{\beta}^{j}
+\delta_{\gamma}^{l}\delta_{\alpha}^{\beta}\delta_{j}^{k}
-\delta_{\beta}^{l}\delta_{\alpha}^{j}\delta_{\gamma}^{k}+\delta_{\beta}^{l}\delta_{\alpha}^{k}\delta_{\gamma}^{j}
+\delta_{\beta}^{l}\delta_{\alpha}^{\gamma}\delta_{j}^{k}  \Big){\rm{Tr}}(\rm{Id})
\nonumber\\
 =&\sum_{1 \leq  j,k,l\leq n}T(e_{j},e_{k},e_{l})\Big(\sum_{ \sigma\in A(3)}sgn(\sigma)
 \sigma(e_{1}^{*}\otimes e_{2}^{*}\otimes e_{3}^{*})  u_{j}v_{k}w_{l}\Big){\rm{Tr}}(\rm{Id})\nonumber\\
  =&T(u,v,w){\rm{Tr}}(\rm{Id}),
\end{align}
and the proof of
the lemma is complete.
\end{proof}
We next make  argument under integral conditions to show  the calculation of traces in several special cases.
According to the integral formula of Case aII in \cite{YW1}, we have
 \begin{equation}
\int_{|\xi|=1}\xi_{i}\xi_{j}\sigma(\xi)=\frac{1}{n}\delta_{i}^{j}{\rm{vol}}(S^{n-1}).
\end{equation}
 \begin{lem}
In terms of the condition (4.17),
the following identities hold:
 \begin{align}
 &\int_{|\xi|=1}{\rm{Tr}}\Big(\sum_{i,j}g^{ij}c(u)c(v)c(w)c(Y)c(\partial_{i})\xi_{j}c(\xi)
+\sum_{i,j}g^{ij}c(u)c(v)c(w)c(\partial_{i})c(Y)\xi_{j}c(\xi) \Big)\sigma(\xi)\nonumber\\
&=- \frac{1}{m}\rm{vol}(S^{2m-1})[g(v,w)g(u,Y)-g(u,w)g(v,Y)+g(u,v)g(w,Y)]{\rm{Tr}}(\rm{Id});\\
 &\int_{|\xi|=1}{\rm{Tr}}\Big(\sum_{i,j}g^{ij}c(u)c(v)c(w)c(T)c(\partial_{i})\xi_{j}c(\xi) \Big)\sigma(\xi)
 =-T(u,v,w){\rm{Tr}}(\rm{Id}){\rm{vol}}(S^{n-1});\\
 &\int_{|\xi|=1}{\rm{Tr}}\Big(\sum_{i,j}g^{ij}c(u)c(v)c(w)c(\partial_{i})c(T)\xi_{j}c(\xi) \Big)\sigma(\xi)
 =-5T(u,v,w){\rm{Tr}}(\rm{Id}){\rm{vol}}(S^{n-1}).
\end{align}
\end{lem}
\begin{proof}
By the relation of the Clifford action and $ {\rm{Tr}}(AB)= {\rm{Tr}}(BA) $, in terms of the condition (4.17), we get
 \begin{align}
  &\int_{|\xi|=1}{\rm{Tr}}\Big(\sum_{i,j}g^{ij}c(u)c(v)c(w)c(Y)c(\partial_{i})\xi_{j}c(\xi)
+\sum_{i,j}g^{ij}c(u)c(v)c(w)c(\partial_{i})c(Y)\xi_{j}c(\xi) \Big)\sigma(\xi)\nonumber\\
=&  \int_{|\xi|=1}{\rm{Tr}}\Big(\sum_{i}c(u)c(v)c(w)c(Y)c(\partial_{i})\xi_{i}c(\xi)
+\sum_{i}c(u)c(v)c(w)c(\partial_{i})c(Y)\xi_{i}c(\xi) \Big)\sigma(\xi)\nonumber\\
=&  \int_{|\xi|=1}\sum_{i,l}\xi_{i}\xi_{l}{\rm{Tr}}\Big(c(u)c(v)c(w)c(Y)c(\partial_{i})c(e_{l})
+c(u)c(v)c(w)c(\partial_{i})c(Y)c(e_{l}) \Big)\sigma(\xi)\nonumber\\
=&-2\int_{|\xi|=1}\sum_{i,l}\xi_{i}\xi_{l}{\rm{Tr}}\Big(c(u)c(v)c(w)g(Y,\partial_{i})c(e_{l}) \Big)\sigma(\xi)\nonumber\\
=&-2\sum_{i}\frac{1}{2m}\delta _{i}^{l}{\rm{vol}}(S^{n-1}){\rm{Tr}}\Big(c(u)c(v)c(w)g(Y,\partial_{i})c(e_{l}) \Big) \nonumber\\
=&-\sum_{i}\frac{1}{m}{\rm{vol}}(S^{n-1}){\rm{Tr}}\Big(c(u)c(v)c(w)g(Y,\partial_{i})c(e_{i}) \Big) \nonumber\\
=&- \frac{1}{m}{\rm{vol}}(S^{n-1}){\rm{Tr}}\Big(c(u)c(v)c(w)c(Y) \Big) \nonumber\\
 =&- \frac{1}{m}{\rm{vol}}(S^{n-1})[g(v,w)g(u,Y)-g(u,w)g(v,Y)+g(u,v)g(w,Y)]{\rm{Tr}}(\rm{Id}).
\end{align}
Applying the trace property and (4.17) again, we get
 \begin{align}
 &\int_{|\xi|=1}{\rm{Tr}}\Big(\sum_{i,j}g^{ij}c(u)c(v)c(w)c(T)c(\partial_{i})\xi_{j}c(\xi) \Big)\sigma(\xi)\nonumber\\
 =&\int_{|\xi|=1}{\rm{Tr}}\Big(\sum_{i,j}g^{ij}c(u)c(v)c(w)c(T)c(\partial_{i})\xi_{j}\sum_{l}\xi_{l}c(e_{l}) \Big)\sigma(\xi)\nonumber\\
   =&\int_{|\xi|=1} {\rm{Tr}}\Big(\sum_{i,l}\xi_{i}\xi_{l}  c(u)c(v)c(w)c(T)c(\partial_{i})c(e_{l}) \Big)\sigma(\xi)\nonumber\\
  =&\sum_{i}\delta _{i}^{l}\frac{1}{n}{\rm{vol}}(S^{n-1})
  {\rm{Tr}}\Big(c(u)c(v)c(w)c(T)c(\partial_{i})c(e_{l}) \Big)\nonumber\\
   =&-{\rm{vol}}(S^{n-1}){\rm{Tr}}\Big( c(u)c(v)c(w)c(T) \Big)\nonumber\\
 =&-T(u,v,w){\rm{Tr}}(\rm{Id}){\rm{vol}}(S^{n-1}).
\end{align}
Using the relation of the Clifford action and (4.22) gives
\begin{align}
&\int_{|\xi|=1}{\rm{Tr}}\Big(\sum_{i,j}g^{ij}c(u)c(v)c(w)c(\partial_{i})c(T)\xi_{j}c(\xi) \Big)\sigma(\xi)\nonumber\\
 =&\int_{|\xi|=1}\sum_{i,m}\xi_{i}\xi_{m}\sum_{1 \leq  h,k,l\leq n}
 \sum_{\alpha,\beta,\gamma=1}^{n}T(e_{h},e_{k},e_{l})u_{\alpha}v_{\beta}w_{\gamma}\nonumber\\
&\times{\rm{Tr}}\big(c(e_{\alpha})c(e_{\beta})c(e_{\gamma})c(e_{i})c(e_{h})c(e_{k})c(e_{l})c(e_{m})\big)\sigma(\xi)\nonumber\\
 =& \sum_{i} \delta _{i}^{m} \frac{1}{n}{\rm{vol}}(S^{n-1})
 \sum_{\alpha,\beta,\gamma=1}^{n}T(e_{h},e_{k},e_{l})u_{\alpha}v_{\beta}w_{\gamma}\nonumber\\
&\times{\rm{Tr}}\Big(c(e_{\alpha})c(e_{\beta})\big(-c(e_{i})c(e_{\gamma})-2\delta_{\gamma}^{i}\big)c(e_{h})c(e_{k})c(e_{l})c(e_{m})\Big)
 \nonumber\\
=&\cdots   \nonumber\\
=&{\rm{vol}}(S^{n-1}) \sum_{1 \leq  h,k,l\leq n}\sum_{\alpha,\beta,\gamma=1}^{n}
T(e_{h},e_{k},e_{l})u_{\alpha}v_{\beta}w_{\gamma}\nonumber\\
&\times{\rm{Tr}}\big(c(e_{\alpha})c(e_{\beta})c(e_{\gamma})c(e_{h})c(e_{k})c(e_{l})\big) \nonumber\\
&-{\rm{vol}}(S^{n-1}) \sum_{1 \leq  h,k,l\leq n}\sum_{\alpha,\beta,\gamma=1}^{n}2T(e_{h},e_{k},e_{l})u_{\alpha}v_{\beta}w_{\gamma}
\delta_{\alpha}^{i}\nonumber\\
&\times{\rm{Tr}}\big( c(e_{\beta})c(e_{\gamma})c(e_{h})c(e_{k})c(e_{l})c(e_{i})\big) \nonumber\\
&+{\rm{vol}}(S^{n-1}) \sum_{1 \leq  h,k,l\leq n}\sum_{\alpha,\beta,\gamma=1}^{n}2T(e_{h},e_{k},e_{l})u_{\alpha}v_{\beta}w_{\gamma}
\delta_{\beta}^{i}\nonumber\\
&\times{\rm{Tr}}\big(c(e_{\alpha}) c(e_{\gamma})c(e_{h})c(e_{k})c(e_{l})c(e_{i})\big)\nonumber\\
&-{\rm{vol}}(S^{n-1}) \sum_{1 \leq  h,k,l\leq n}\sum_{\alpha,\beta,\gamma=1}^{n}2T(e_{h},e_{k},e_{l})u_{\alpha}v_{\beta}w_{\gamma}
\delta_{\gamma}^{i}\nonumber\\
&\times{\rm{Tr}}\big(c(e_{\alpha})c(e_{\beta}) c(e_{h})c(e_{k})c(e_{l})c(e_{i})\big) \nonumber\\
 =&-5T(u,v,w){\rm{Tr}}(\rm{Id}){\rm{vol}}(S^{n-1}).
\end{align}
Then the proof of the lemma is complete.
\end{proof}
From (3.10), Lemma 4.3 and Lemma 4.4, we obtain the main theorem in this section.
\begin{thm}
For Connes operator with the trilinear Clifford multiplication by functional of differential one-forms
$c(u),c(v),c(w)$, the spectral torsion for  type-I operator  $D_{F}+(c(T)+\sqrt{-1}c(Y)) \otimes \Phi $  equals to
\begin{align}
&\mathscr{T}(c(u),c(v),c(w))\nonumber\\
=&\int_M\int_{|\xi|=1}{\rm Tr}_{S(TM)\otimes F}
[\sigma_{-2m}\big(c(u)c(v)c(w)(D_{F}+(c(T)+\sqrt{-1}c(Y))\otimes \Phi )^{-2m+1}\big)]\sigma(\xi){\rm d}x\nonumber\\
=&\int_M -2^{m+1}T(u,v,w)
{\rm{ Tr}}^{F}( \Phi)\rm{vol}(S^{n-1}) {\rm d}\rm{vol}_{M}.
\end{align}
\end{thm}
\begin{rem}
In the preceding the calculation of the Wodzicki residue
research to recover the torsion of the linear
connection,  Dabrowski et al. \cite{DSZ} showed that the spectral definition of torsion can be readily extended to the noncommutative
case of spectral triples. In this section, the computation of  spectral torsion for  Connes operator  $D_{F}+(c(T)+\sqrt{-1}c(Y)) \otimes \Phi $
is the elaboration of spectral torsion in \cite{DSZ}.
\end{rem}
Substituting $T=0$ into above calculation we get
\begin{rem}
For the Connes type operator with the trilinear Clifford multiplication by functional of differential one-forms
$c(u),c(v),c(w)$, the spectral torsion for the Connes type operator  $D_{F}+\sqrt{-1}c(Y)\otimes \Phi $  equals to
\begin{align}
\mathscr{T}(c(u),c(v),c(w))
=\int_M\int_{|\xi|=1}{\rm Tr}_{S(TM)\otimes F}
[\sigma_{-2m}\big(c(u)c(v)c(w)(D_{F}+\sqrt{-1}c(Y)\otimes \Phi )^{-2m+1}\big)]\sigma(\xi){\rm d}x=0.
\end{align}
\end{rem}

\subsection{Spectral torsion for the  Connes type operator with twisted fluctuation as torsion}
In this section we show the type-II Connes operator with twisted fluctuation as torsion, namely that the minimal twist of
an Riemannian spin manifold induces a orthogonal and geodesic preserving torsion.
By using the technique developed in Section 3 to get the noncommutative residue density
of a suitable power of the Connes type operator, and give the representation
of spectral torsion for the Connes type operator with twisted fluctuation as torsion.

Let us begin by a technical lemma showing that the product of the grading $\gamma$ by any
Euclidean Dirac matrix  results.
The grading operator $\gamma$ denote by
\begin{align}
\gamma= (\sqrt{-1})^{m}\prod_{j=1}^{2m}c(e_{j}).
\end{align}
In the terms of the orthonormal frames $\{e_{i}\}(1\leq i,j\leq n)$ on $TM$, we have
 $
\gamma=(\sqrt{-1})^{m}c(e_{1})c(e_{2})\cdots c(e_{n}).
$
\begin{defn}\cite{Y}
 Suppose $V$ is a super vector space, and $\gamma$ is its super structure. If $\phi\in \rm{End}(V)$,
 let ${\rm{Tr}}(\phi)$ be the trace of $\phi$, then define
 \begin{align}
{\rm{ Str}}(\phi)={\rm{Tr}}(\gamma\circ\phi).
\end{align}
 Here ${\rm{Tr}}$ and $Str$ are called the trace and the super trace of $\phi$ respectively.
 \end{defn}
 \begin{lem}\cite{Y}
The super trace (function) $Str:\rm{End}_{\mathbb{C}}(S(2m))\rightarrow \mathbb{C}$
is a complex linear map satisfying
\begin{eqnarray}
{\rm{ Str}}(c(e_{i_{1}})c(e_{i_{w}})\cdots c(e_{i_{q}})) &=&\left\{
       \begin{array}{c}
        0,  ~~~~~~~~~~~~~~~~~~~~~~~{\rm if }~q<2m; \\[2pt]
      \frac{2^{m}}{(\sqrt{-1})^{m}}, ~~~~~~~~~~~~~~~{\rm if }~q=2m,
       \end{array}
    \right.
\end{eqnarray}
where $1\leq i_{1},i_{2},i_{q}\leq 2m $.
\end{lem}

Now we explore an alternative way, consisting in generating torsion from a
torsionless connection, through a twisted fluctuation of the Connes operator. In addition,
we divide the proof of spectral torsion for the Connes type operator with twisted fluctuation into {\bf Case (1)},{\bf Case (2)} and {\bf Case (3)}.

{\bf Case (1)} The Spectral torsion for $D_{F}+\gamma \otimes \Phi$

First, we compute the symbol of $\sigma_{-n}\big(c(u)c(v)c(w)(D_{F}+\gamma \otimes \Phi )^{-n+1}\big)$.
Substituting $D_{F}+\gamma \otimes \Phi$ into (3.4) yields
 \begin{align}
\Big(D_{F}+\gamma\otimes \Phi \Big)^{2}
=&D_{F}^{2}
+\sum_{j=1}^{n}c(e_{j})\Big(\nabla_{e_{j}}^{ S(TM)}\otimes \rm{id}_{ F}+ \rm{id}_{ S(TM)}\otimes
\nabla^{F}_{e_{j}}\Big)(\gamma\otimes \Phi)\nonumber\\
&+(\gamma\otimes \Phi)\sum_{j=1}^{n}c(e_{j})\Big(\nabla_{e_{j}}^{ S(TM)}\otimes \rm{id}_{ F}+ \rm{id}_{ S(TM)}\otimes \nabla^{F}_{e_{j}}\Big)
+\Phi^{2}\nonumber\\
=&D_{F}^{2}
+\Big(\sum_{j=1}^{n}c(e_{j})\otimes \nabla_{e_{j}}^{F}\Big)(\gamma\otimes \Phi)
+(\gamma\otimes \Phi)\Big(\sum_{j=1}^{n}c(e_{j})\otimes \nabla_{e_{j}}^{F}\Big)
+\Phi^{2}\nonumber\\
=&D_{F}^{2}
+\sum_{j=1}^{n}c(e_{j})\gamma\otimes (\nabla^{F}_{e_{j}}\Phi)
+\sum_{j=1}^{n}\gamma c(e_{j})\otimes \Phi\nabla^{F}_{e_{j}}
+\Phi^{2}\nonumber\\
=&D_{F}^{2}
-\gamma \sum_{j=1}^{n} c(e_{j})\otimes\nabla^{F}_{e_{j}}(\Phi)
+\Phi^{2}.
\end{align}
Then we obtain
$\sigma_{-3}\big((D_{F}+\gamma\otimes \Phi)^{-2}\big)=\sigma_{-3}\big(D_{F} ^{-2}\big)$.
From (4.11) and the main symbol of $D_{F}+\gamma \otimes \Phi$, we get
\begin{align}
 \sigma_{-2m}\big(c(u)c(v)c(w)(D_{F}+\gamma\otimes \Phi)^{1-2m}\big)(x_{0})|_{|\xi|=1}
=-c(u)c(v)c(w)\big(\gamma\otimes \Phi\big).
\end{align}
 We note that ${\rm{Tr}}\big(c(u)c(v)c(w)\gamma\big)={\rm{Str}}\big(c(u)c(v)c(w)\big) =0$, then
\begin{thm}
For the Connes type operator with the trilinear Clifford multiplication by functional of differential one-forms
$c(u),c(v),c(w)$, the spectral torsion for  the Connes type operator  $D_{F}+\gamma\otimes \Phi $  equal to
\begin{eqnarray}
\mathscr{T}(c(u),c(v),c(w)) =0.
\end{eqnarray}
\end{thm}

{\bf Case (2)} The Spectral torsion for $D_{F}+ c(X)\gamma \otimes \Phi $

Now let¡¯s recall the $-2m$ order main symbol of Connes operator, which is
important to the Spectral torsion for $D_{F}+ c(X)\gamma \otimes \Phi $.
 Let $X=\sum_{i=1}^{n}X_{i}e_{i}$, by (3.9), we get
\begin{align}
 &\sigma_{-2m}(c(u)c(v)c(w)\big(D_{F}+ c(X)\gamma \otimes \Phi\big)^{1-2m})(x_{0})|_{|\xi|=1} \nonumber\\
=&c(u)c(v)c(w) c(X)\gamma \otimes \Phi
+mg^{ij}c(u)c(v)c(w)\Big[c(\partial_{i})\big( c(X)\gamma \otimes \Phi\big) \nonumber\\
&+\big( c(X)\gamma\otimes \Phi \big)c(\partial_{i})   \Big]\xi_{j}c(\xi).
\end{align}
Next we compute  the spectral torsion for a twisted fluctuation of the Connes type operator $D_{F}+ c(X)\gamma \otimes \Phi$.
When $n=2m=4$, by lemma 4.9 we get
\begin{align}
{\rm{Tr}}\big(c(u)c(v)c(w)c(X)\gamma )\big)
 =&\sum_{j,k,l,i=1}^{n}  u_{j}v_{k}w_{l}X_{i}{\rm{Tr}}\big(c(e_{j})c(e_{k}) c(e_{l}) c(e_{i})
       c(X)\gamma\big)\nonumber\\
 =&\sum_{i,j,k,l=1}^{n}  u_{j}v_{k}w_{l}X_{i}{\rm{ Str}}\big(c(e_{j})c(e_{k}) c(e_{l}) c(e_{i})\big)\nonumber\\
  =&\sum_{\{i,j,k,l \}=\{1,2,3,4\}}   \frac{2^{2}}{(\sqrt{-1})^{2}}u_{j}v_{k}w_{l}X_{i}\nonumber\\
   =&-4\langle  u^{*}\wedge v^{*}\wedge w^{*}\wedge X^{*},e_{1}^{*}\wedge e_{2}^{*}\wedge e_{3}^{*}\wedge e_{4}^{*}  \rangle.
\end{align}
 Then we obtain
\begin{eqnarray}
{\rm{Tr}}\big(c(u)c(v)c(w)c(X)\gamma \big) &=&\left\{
       \begin{array}{c}
        0,  ~~~~~~~~~~~~~~~~~~~~~~~~~~~~~~~~~~~~~~~~~~~~~~~~~~ ~~{\rm if }~2m\neq4; \\[2pt]
      -4\langle u^{*}\wedge v^{*}\wedge w^{*}\wedge X^{*},e_{1}^{*}\wedge e_{2}^{*}\wedge e_{3}^{*}\wedge e_{4}^{*}  \rangle, ~~{\rm if }~2m=4.
       \end{array}
    \right.
\end{eqnarray}
By the relation of the Clifford action and the computation of  super trace, we get
\begin{align}
&\int_{|\xi|=1}\sum_{i,j}g^{ij}{\rm{Tr}}\big(c(u)c(v)c(w)c(X)\gamma c(\partial_{i})\xi_{j}c(\xi) \big)\sigma(\xi)\nonumber\\
=&\int_{|\xi|=1}\sum_{i}\xi_{i}\xi_{l}{\rm{Tr}}\big(c(u)c(v)c(w)c(X)\gamma c(\partial_{i})c(e_{l}) \big)\sigma(\xi)\nonumber\\
=&\sum_{i}\delta _{i}^{l}\frac{1}{n}{\rm{vol}}(S^{n-1})   {\rm{Tr}}\big(c(u)c(v)c(w)c(X)\gamma c(\partial_{i})c(e_{l}) \big)\nonumber\\
=&- {\rm{vol}}(S^{n-1}) {\rm{Tr}}\big(c(u)c(v)c(w)c(X)\gamma\big).
\end{align}
Also, straightforward computations yield
\begin{align}
&\int_{|\xi|=1}\sum_{i,j}g^{ij}{\rm{Tr}}\big(c(u)c(v)c(w)c(\partial_{i})c(X)\gamma\xi_{j}c(\xi) \big)\sigma(\xi)\nonumber\\
=&\int_{|\xi|=1}\sum_{i,l} \xi_{i}\xi_{l}{\rm{Tr}}\big(c(u)c(v)c(w)c(\partial_{i})c(X)\gamma c(\partial_{l}) \big)\sigma(\xi)\nonumber\\
=&-\sum_{i,l }\delta _{i}^{l}  \frac{1}{n} {\rm{vol}}(S^{n-1}) {\rm{Tr}}\big(c(u)c(v)c(w)c(\partial_{i})c(X) c(\partial_{l})\gamma \big) \nonumber\\
=&-\sum_{i }  \frac{1}{n} {\rm{vol}}(S^{n-1})  {\rm{Tr}}\Big(c(u)c(v)c(w)
\big(-c(X)c(\partial_{i})-2g(\partial_{i},X )\big) c(\partial_{i}) \gamma\Big) \nonumber\\
=&- \sum_{i }  \frac{1}{n} {\rm{vol}}(S^{n-1})
  {\rm{Tr}}\big(c(u)c(v)c(w)c(X) \gamma\big) \nonumber\\
&+ \sum_{i } \frac{1}{n}{\rm{vol}}(S^{n-1})  {\rm{Tr}}\Big(c(u)c(v)c(w)
\big(2g(\partial_{i},X_{i}) \big)c(e_{i})\gamma\Big) \nonumber\\
=&-{\rm{vol}}(S^{n-1})  {\rm{Tr}}\Big(c(u)c(v)c(w)
c(X)\gamma\Big)
+\frac{2}{n}{\rm{vol}}(S^{n-1})  {\rm{Tr}}\Big(c(u)c(v)c(w)c(X)\gamma\Big)\nonumber\\
=&\frac{2-n}{n}{\rm{vol}}(S^{n-1})  {\rm{Tr}}\big(c(u)c(v)c(w)c(X)\gamma\big).
\end{align}
Substituting into above calculation, we obtain
 \begin{thm}
For the Connes type  operator with the trilinear Clifford multiplication by functional of differential one-forms
$c(u),c(v),c(w)$, the spectral torsion for the Connes type operator  $D_{F}+ c(X)\gamma\otimes \Phi $  equal to
\begin{eqnarray}
\mathscr{T}(c(u),c(v),c(w)) &=&\left\{
       \begin{array}{c}
        0,  ~~~~~~~~~~~~~~~~~~~~~~~~~~~~~~~~~~~~~~~~~~~~~~~~~~~~~~~~~~~~~~~{\rm if }~ 2m > 4; \\[2pt]
      \int_M  2^{3}(u^{*}\wedge v^{*}\wedge w^{*}\wedge X^{*}){\rm{vol}}(S^{3}), ~~~~~~~~~~~~~~~~~~~~~~~~{\rm if }~2m=4.
       \end{array}
    \right.
\end{eqnarray}
\end{thm}

{\bf Case (3)} The Spectral torsion for $D_{F}+ \sqrt{-1}c(T)\gamma \otimes \Phi $

Just like in Subsection 4.1, we now consider the the spectral torsion for $D_{F}+ \sqrt{-1}c(T)\gamma \otimes \Phi $
which is an analogy of the representation of such operator and the noncommutative residue density.
 The difference is that here we use the super trace of $D_{F}+ \sqrt{-1}c(T)\gamma \otimes \Phi $, instead of the trace of $D_{F}+(c(T)+\sqrt{-1}c(Y))\otimes \Phi$.
By (3.9), we get
\begin{align}
 &\sigma_{-2m}(c(u)c(v)c(w)\big(D_{F}+ \sqrt{-1}c(T)\gamma \otimes \Phi\big)^{1-2m})(x_{0})|_{|\xi|=1} \nonumber\\
=&c(u)c(v)c(w) \sqrt{-1}c(T)\gamma \otimes \Phi
+mg^{ij}c(u)c(v)c(w)\Big[c(\partial_{i})\big( \sqrt{-1}c(T)\gamma \otimes \Phi\big) \nonumber\\
&+\big( \sqrt{-1}c(T)\gamma\otimes \Phi \big)c(\partial_{i})   \Big]\xi_{j}c(\xi).
\end{align}
Next we compute  the super trace of such operators $D_{F}+ \sqrt{-1}c(T)\gamma \otimes \Phi$ and the noncommutative residue density.

(1). When $n=2m=6$, by the relation of the Clifford action and lemma 4.9, we obtain
\begin{align}
{\rm{Tr}}\big(c(u)c(v)c(w)c(T)\gamma )\big)
=&\sum_{1 \leq  j,k,l\leq n} T(e_{j},e_{k},e_{l})
{\rm{Tr}}\big(c(u)c(v)c(w)c(e_{j})c(e_{k})c(e_{l})\gamma\big)\nonumber\\
=&\sum_{1 \leq  j,k,l\leq n} T(e_{j},e_{k},e_{l})
{\rm{ Str}}\big(c(u)c(v)c(w)c(e_{j})c(e_{k})c(e_{l})\big)\nonumber\\
=&\frac{2^{3}}{(\sqrt{-1})^{3}}\langle  u^{*}\wedge v^{*}\wedge w^{*}\wedge T,
e_{1}^{*}\wedge e_{2}^{*}\wedge e_{3}^{*}\wedge e_{4}^{*}\wedge e_{5}^{*}\wedge e_{6}^{*}  \rangle.
\end{align}
By (4.39) and the computation of  super trace, we get
\begin{align}
{\rm{Tr}}\big(c(u)c(v)c(w)c(T)\gamma \big) =&\left\{
       \begin{array}{c}
        0,  ~~~~~~~~~~~~~~~~~~~~~~~~~~~~~~~~~~~~~~~~~~~~~~~~~~~~~~~~~~~~~~~~~~ ~~{\rm if }~2m>6; \\[2pt]
     \frac{2^{3}}{(\sqrt{-1})^{3}}\langle  u^{*}\wedge v^{*}\wedge w^{*}\wedge T,
e_{1}^{*}\wedge e_{2}^{*}\wedge e_{3}^{*}\wedge e_{4}^{*}\wedge e_{5}^{*}\wedge e_{6}^{*}  \rangle, ~~{\rm if }~2m=6.
       \end{array}
    \right.
\end{align}
In terms of the condition (4.17), we get
\begin{align}
&\int_{|\xi|=1}\sum_{i,j}g^{ij}{\rm{Tr}}\big(c(u)c(v)c(w)c(\partial_{i})c(T)\gamma\xi_{j}c(\xi) \big)\sigma(\xi)\nonumber\\
=&\int_{|\xi|=1}\sum_{i,l} \xi_{i}\xi_{l}{\rm{Tr}}\big(c(u)c(v)c(w)c(\partial_{i})c(T)\gamma c(\partial_{l}) \big)\sigma(\xi)\nonumber\\
=&\sum_{i,l }\delta _{i}^{l}  \frac{1}{n} {\rm{vol}}(S^{n-1}){\rm{Tr}}\big(c(u)c(v)c(w)c(\partial_{i})c(T)\gamma c(\partial_{l}) \big) \nonumber\\
=&\sum_{i }  \frac{1}{n} {\rm{vol}}(S^{n-1})  {\rm{Tr}}\Big(c(u)c(v)
\big(-c(\partial_{i})c(w)-2g(w, \partial_{i})\big) c(T)\gamma c(\partial_{i})\Big) \nonumber\\
=&- \sum_{i }  \frac{1}{n} {\rm{vol}}(S^{n-1})
  {\rm{Tr}}\big(c(u)c(v)c(\partial_{i})c(w)c(T)\gamma c(\partial_{i})  \big) \nonumber\\
  &- \frac{2}{n}{\rm{vol}}(S^{n-1})  {\rm{Tr}}\Big(c(u)c(v)c(T)\gamma c(w)\Big) \nonumber\\
=& \sum_{i }  \frac{1}{n} {\rm{vol}}(S^{n-1})
  {\rm{Tr}}\big(c(u)c(\partial_{i})c(v)c(w)c(T)\gamma c(\partial_{i})  \big)
   + \frac{2}{n}{\rm{vol}}(S^{n-1})  {\rm{Tr}}\Big(c(u)c(w)c(T)\gamma c(v)\Big) \nonumber\\
  &- \frac{2}{n}{\rm{vol}}(S^{n-1})  {\rm{Tr}}\Big(c(u)c(v)c(T)\gamma c(w)\Big) \nonumber\\
=&   {\rm{vol}}(S^{n-1})
  {\rm{Tr}}\big( c(u)c(v)c(w)c(T)\gamma    \big)
   - \frac{2}{n}{\rm{vol}}(S^{n-1})  {\rm{Tr}}\Big(c(v)c(w)c(T)\gamma c(u)\Big) \nonumber\\
  &+ \frac{2}{n}{\rm{vol}}(S^{n-1})  {\rm{Tr}}\Big(c(u)c(w)c(T)\gamma c(v)\Big)
 - \frac{2}{n}{\rm{vol}}(S^{n-1})  {\rm{Tr}}\Big(c(u)c(v)c(T)\gamma c(w)\Big) \nonumber\\
 =&\frac{n-6}{n}{\rm{vol}}(S^{n-1})
  {\rm{Tr}}\big( c(u)c(v)c(w)c(T)\gamma    \big).
\end{align}
Also, straightforward computations yield
\begin{align}
&\int_{|\xi|=1}\sum_{i,j}g^{ij}{\rm{Tr}}\big(c(u)c(v)c(w)c(T)\gamma c(\partial_{i})\xi_{j}c(\xi) \big)\sigma(\xi)\nonumber\\
=&\int_{|\xi|=1}\sum_{i,l} \xi_{i}\xi_{l}{\rm{Tr}}\big(c(u)c(v)c(w)c(T)\gamma c(\partial_{i})c(\partial_{l}) \big)\sigma(\xi)\nonumber\\
=&\sum_{i,l }\delta _{i}^{l}  \frac{1}{n} {\rm{vol}}(S^{n-1}){\rm{Tr}}\big(c(u)c(v)c(w)c(T)\gamma c(\partial_{i})c(\partial_{l}) \big) \nonumber\\
 =&-{\rm{vol}}(S^{n-1})  {\rm{Tr}}\big( c(u)c(v)c(w)c(T)\gamma \big) .
\end{align}
Then $n=2m=6$, substituting above calculation into (3.10)  leads to the Spectral torsion for $D_{F}+ \sqrt{-1}c(T)\gamma \otimes \Phi $
\begin{align}
\mathscr{T}(c(u),c(v),c(w)) =  \int_M  2^{4}(u^{*}\wedge v^{*}\wedge w^{*}\wedge T){\rm{Tr}}^{F}(\Phi){\rm{vol}}(S^{5}).
\end{align}

(2). When $n=2m=4$, by lemma 4.9, we have the equality:
\begin{align}
&{\rm{Tr}}\big(c(u)c(v)c(w)c(T)\gamma )\big)\nonumber\\
=&\sum_{1 \leq  j,k,l\leq 4} T(e_{j},e_{k},e_{l})
{\rm{Tr}}\big(c(u)c(v)c(w)c(e_{j})c(e_{k})c(e_{l})\gamma\big)\nonumber\\
 =&\sum_{1 \leq  j,k,l\leq 4} T(e_{j},e_{k},e_{l})
{\rm{Tr}}\big(\gamma c(u)c(v)c(w)c(e_{j})c(e_{k})c(e_{l})\big)\nonumber\\
 =& {\rm{Tr}}\Big[\gamma \sum_{\alpha,\beta,\gamma=1}^{4}u_{\alpha}v_{\beta}w_{\gamma}  c(e_{\alpha})c(e_{\beta})c(e_{\gamma})
  \Big( T(e_{1},e_{2},e_{3}) c(e_{1})c(e_{2})c(e_{3})+T(e_{1},e_{2},e_{4}) c(e_{1})c(e_{2})c(e_{4})\nonumber\\
  &+T(e_{1},e_{3},e_{4}) c(e_{1})c(e_{3})c(e_{4})+T(e_{2},e_{3},e_{4}) c(e_{2})c(e_{3})c(e_{4})
    \Big) \Big]\nonumber\\
   =& {\rm{Str}}\Big[\sum_{\alpha,\beta,\gamma=1}^{4}u_{\alpha}v_{\beta}w_{\gamma}  c(e_{\alpha})c(e_{\beta})c(e_{\gamma})
  \Big( T(e_{1},e_{2},e_{3}) c(e_{1})c(e_{2})c(e_{3})+T(e_{1},e_{2},e_{4}) c(e_{1})c(e_{2})c(e_{4})\nonumber\\
  &+T(e_{1},e_{3},e_{4}) c(e_{1})c(e_{3})c(e_{4})+T(e_{2},e_{3},e_{4}) c(e_{2})c(e_{3})c(e_{4})
    \Big) \Big]\nonumber\\
 =& H_{1}|_{\alpha=\beta}+H_{2}|_{\alpha=\gamma}+H_{3 }|_{\beta=\gamma}-2 H_{4 }|_{\alpha=\beta=\gamma}.
  \end{align}
By the relation of the Clifford action and $ {\rm{Tr}}(AB)= {\rm{Tr}}(BA) $, we obtain
 \begin{align}
    & H_{1}|_{\alpha=\beta}\nonumber\\
    =& {\rm{Str}}\Big[\sum_{\alpha,\gamma=1}^{4}u_{\alpha}v_{\alpha}w_{\gamma}  c(e_{\alpha})c(e_{\alpha})c(e_{\gamma})
  \Big( T(e_{1},e_{2},e_{3}) c(e_{1})c(e_{2})c(e_{3})+T(e_{1},e_{2},e_{4}) c(e_{1})c(e_{2})c(e_{4})\nonumber\\
  &+T(e_{1},e_{3},e_{4}) c(e_{1})c(e_{3})c(e_{4})+T(e_{2},e_{3},e_{4}) c(e_{2})c(e_{3})c(e_{4})
    \Big) \Big]\nonumber\\
     =& {\rm{Str}}\Big[-g(u,v)
  \Big(\big( -T(e_{1},e_{2},e_{3}) w_{4} +T(e_{1},e_{2},e_{4}) w_{3}
  -T(e_{1},e_{3},e_{4}) w_{2}+T(e_{2},e_{3},e_{4})w_{1}  \big)\nonumber\\
  &\times c(e_{1}) c(e_{2})c(e_{3})c(e_{4})\Big) \Big].
  \end{align}
In the similar way we get
 \begin{align}
   &H_{2}|_{\alpha=\gamma}\nonumber\\
    =& {\rm{Str}}\Big[\sum_{\alpha,\beta=1}^{4}u_{\alpha}v_{\beta}w_{\alpha}  c(e_{\alpha})c(e_{\beta})c(e_{\alpha})
  \Big( T(e_{1},e_{2},e_{3}) c(e_{1})c(e_{2})c(e_{3})+T(e_{1},e_{2},e_{4}) c(e_{1})c(e_{2})c(e_{4})\nonumber\\
  &+T(e_{1},e_{3},e_{4}) c(e_{1})c(e_{3})c(e_{4})+T(e_{2},e_{3},e_{4}) c(e_{2})c(e_{3})c(e_{4})
    \Big) \Big]\nonumber\\
  =& {\rm{Str}}\Big[\sum_{\alpha,\beta=1}^{4}u_{\alpha}v_{\beta}w_{\alpha}
        \big(-c(e_{\beta})c(e_{\alpha}) -2\delta_{\alpha}^{\beta}\big)c(e_{\alpha})
  \Big( T(e_{1},e_{2},e_{3}) c(e_{1})c(e_{2})c(e_{3})\nonumber\\
  &+T(e_{1},e_{2},e_{4}) c(e_{1})c(e_{2})c(e_{4})+T(e_{1},e_{3},e_{4}) c(e_{1})c(e_{3})c(e_{4})+T(e_{2},e_{3},e_{4}) c(e_{2})c(e_{3})c(e_{4})
    \Big) \Big]\nonumber\\
     =& {\rm{Str}}\Big[g(u,w)
  \Big(\big( -T(e_{1},e_{2},e_{3}) v_{4} +T(e_{1},e_{2},e_{4}) v_{3}
  -T(e_{1},e_{3},e_{4}) v_{2}+T(e_{2},e_{3},e_{4})v_{1}  \big)\nonumber\\
  &\times c(e_{1}) c(e_{2})c(e_{3})c(e_{4})\Big) \Big]\nonumber\\
  & -2{\rm{Str}}\Big[\sum_{\alpha =1}^{4}u_{\alpha}v_{\alpha}w_{\alpha}
         c(e_{\alpha})  \Big( T(e_{1},e_{2},e_{3}) c(e_{1})c(e_{2})c(e_{3})+T(e_{1},e_{2},e_{4}) c(e_{1})c(e_{2})c(e_{4})\nonumber\\
  &+T(e_{1},e_{3},e_{4}) c(e_{1})c(e_{3})c(e_{4})+T(e_{2},e_{3},e_{4}) c(e_{2})c(e_{3})c(e_{4})
    \Big) \Big].
       \end{align}
Also, straightforward computations yield
 \begin{align}
&H_{3 }|_{\beta=\gamma}\nonumber\\
    =& {\rm{Str}}\Big[\sum_{\alpha,\beta=1}^{4}u_{\alpha}v_{\beta}w_{\beta}  c(e_{\alpha})c(e_{\beta})c(e_{\beta})
  \Big( T(e_{1},e_{2},e_{3}) c(e_{1})c(e_{2})c(e_{3})+T(e_{1},e_{2},e_{4}) c(e_{1})c(e_{2})c(e_{4})\nonumber\\
  &+T(e_{1},e_{3},e_{4}) c(e_{1})c(e_{3})c(e_{4})+T(e_{2},e_{3},e_{4}) c(e_{2})c(e_{3})c(e_{4})
    \Big) \Big]\nonumber\\
     =& {\rm{Str}}\Big[-g(v,w)
  \Big(\big( -T(e_{1},e_{2},e_{3}) u_{4} +T(e_{1},e_{2},e_{4}) u_{3}
  -T(e_{1},e_{3},e_{4}) u_{2}+T(e_{2},e_{3},e_{4})u_{1}  \big)\nonumber\\
  &\times c(e_{1}) c(e_{2})c(e_{3})c(e_{4})\Big) \Big],
\end{align}
and
 \begin{align}
&H_{4}|_{\alpha=\beta=\gamma}\nonumber\\
    =& {\rm{Str}}\Big[\sum_{\alpha=1}^{4}u_{\alpha}v_{\alpha}w_{\alpha}  c(e_{\alpha})c(e_{\alpha})c(e_{\alpha})
  \Big( T(e_{1},e_{2},e_{3}) c(e_{1})c(e_{2})c(e_{3})+T(e_{1},e_{2},e_{4}) c(e_{1})c(e_{2})c(e_{4})\nonumber\\
  &+T(e_{1},e_{3},e_{4}) c(e_{1})c(e_{3})c(e_{4})+T(e_{2},e_{3},e_{4}) c(e_{2})c(e_{3})c(e_{4})
    \Big) \Big] \nonumber\\
     =& -{\rm{Str}}\Big[\sum_{\alpha=1}^{4}u_{\alpha}v_{\alpha}w_{\alpha}  c(e_{\alpha})
  \Big( T(e_{1},e_{2},e_{3}) c(e_{1})c(e_{2})c(e_{3})+T(e_{1},e_{2},e_{4}) c(e_{1})c(e_{2})c(e_{4})\nonumber\\
  &+T(e_{1},e_{3},e_{4}) c(e_{1})c(e_{3})c(e_{4})+T(e_{2},e_{3},e_{4}) c(e_{2})c(e_{3})c(e_{4})
    \Big) \Big]
\end{align}
Combining  (4.45)-(4.48),  we obtain
\begin{align}
&{\rm{Tr}}\big(c(u)c(v)c(w)c(T)\gamma )\big)\nonumber\\
  =& H_{1}|_{\alpha=\beta}+H_{2}|_{\alpha=\gamma}+H_{3 }|_{\beta=\gamma}-2 H_{4 }|_{\alpha=\beta=\gamma}\nonumber\\
 =&-g(u,v) \big( -T(e_{1},e_{2},e_{3}) w_{4} +T(e_{1},e_{2},e_{4}) w_{3}
  -T(e_{1},e_{3},e_{4}) w_{2}\nonumber\\
    &+T(e_{2},e_{3},e_{4})w_{1}  \big) {\rm{Str}}\big( c(e_{1}) c(e_{2})c(e_{3})c(e_{4})\big)\nonumber\\
 &+g(u,w) \big( -T(e_{1},e_{2},e_{3}) v_{4} +T(e_{1},e_{2},e_{4}) v_{3}
  -T(e_{1},e_{3},e_{4}) v_{2}\nonumber\\
  &+T(e_{2},e_{3},e_{4})v_{1}  \big) {\rm{Str}}\big( c(e_{1}) c(e_{2})c(e_{3})c(e_{4})\big)\nonumber\\
  &-g(v,w) \big( -T(e_{1},e_{2},e_{3}) u_{4} +T(e_{1},e_{2},e_{4}) u_{3}
  -T(e_{1},e_{3},e_{4}) u_{2}\nonumber\\
  &+T(e_{2},e_{3},e_{4})u_{1}  \big) {\rm{Str}}\big( c(e_{1}) c(e_{2})c(e_{3})c(e_{4})\big)\nonumber\\
  =&-g(u,v) \frac{2^{2}}{(\sqrt{-1})^{2}}\langle   w^{*}\wedge T,
e_{1}^{*}\wedge e_{2}^{*}\wedge e_{3}^{*}\wedge e_{4}^{*}   \rangle
+g(u,w) \frac{2^{2}}{(\sqrt{-1})^{2}}\langle   v^{*}\wedge T,
e_{1}^{*}\wedge e_{2}^{*}\wedge e_{3}^{*}\wedge e_{4}^{*}   \rangle\nonumber\\
&-g(v,w) \frac{2^{2}}{(\sqrt{-1})^{2}}\langle  u^{*}\wedge T,
e_{1}^{*}\wedge e_{2}^{*}\wedge e_{3}^{*}\wedge e_{4}^{*}   \rangle\nonumber\\
=&\frac{2^{2}}{(\sqrt{-1})^{2}} \langle -g(u,v)w^{*}\wedge T+ g(u,w)v^{*}\wedge T-g(v,w)u^{*}\wedge T,
e_{1}^{*}\wedge e_{2}^{*}\wedge e_{3}^{*}\wedge e_{4}^{*}   \rangle.
    \end{align}
 Then $n=2m=4$, substituting (4.48) into (3.10)  leads to the Spectral torsion for $D_{F}+ \sqrt{-1}c(T)\gamma \otimes \Phi $
\begin{align}
\mathscr{T}(c(u),c(v),c(w)) =&
 \int_M  2^{4}\sqrt{-1}\langle -g(u,v)w^{*}\wedge T+ g(u,w)v^{*}\wedge T- g(v,w)u^{*}\wedge T,
e_{1}^{*}\wedge e_{2}^{*}\wedge e_{3}^{*}\wedge e_{4}^{*}   \rangle \nonumber\\
&\times {\rm{Tr}}^{F}(\Phi){\rm{vol}}(S^{3}){\rm{d}}x \nonumber\\
=& \int_M  2^{4}\sqrt{-1}(-g(u,v)w^{*}\wedge T+ g(u,w)v^{*}\wedge T-g(v,w)u^{*}\wedge T) {\rm{Tr}}^{F}(\Phi){\rm{vol}}(S^{3}).
\end{align}
From  (4.43) and (4.49),  we have
 \begin{thm}
For the Connes type  operator with the trilinear Clifford multiplication by functional of differential one-forms
$c(u),c(v),c(w)$, the spectral torsion for the Connes type operator  $D_{F}+ \sqrt{-1}c(T)\gamma \otimes \Phi $
equals to
\begin{align}
\mathscr{T}(c(u),c(v),c(w)) =&\left\{
       \begin{array}{c}
        \int_M  2^{4}\sqrt{-1}(-g(u,v)w^{*}\wedge T+ g(u,w)v^{*}\wedge T-g(v,w)u^{*}\wedge T) {\rm{Tr}}^{F}(\Phi){\rm{vol}}(S^{3}),
           ~n=4; \\[2pt]
       \int_M  2^{4}(u^{*}\wedge v^{*}\wedge w^{*}\wedge T){\rm{Tr}}^{F}(\Phi){\rm{vol}}(S^{5}),  ~~~~~~~~~~~~~~~~~~~~~~~~~~~~~~~~~~~~~~~~~~~n=6; \\[2pt]
     0, ~~~~~~~~~~~~~~~~~~~~~~~~~~~~~~~~~~~~~~~~~~~~~~~~~~~~~~~~~~~~~~~~~~~~~~~~~~~~~~~~~~n\neq 4,n\neq 6.
       \end{array}
    \right.
\end{align}
\end{thm}

\subsection{The spectral torsion for manifold with boundary associated with $D_{F}+(c(T)+\sqrt{-1}c(Y)) \otimes \Phi $  }
The purpose of this section is to specify The spectral torsion on manifold with boundary  for the Connes type operator with torsion
$\widetilde{D}_{F}=D_{F}+(c(T)+\sqrt{-1}c(Y)) \otimes \Phi $ .
Now we recall the main theorem in \cite{FGLS}.
\begin{thm}\cite{FGLS}
 Let $X$ and $\partial X$ be connected, ${\rm dim}X=n\geq3$,
 $A=\left(\begin{array}{lcr}\pi^+P+G &   K \\
T &  S    \end{array}\right)$ $\in \mathcal{B}$ , and denote by $p$, $b$ and $s$ the local symbols of $P,G$ and $S$ respectively.
 Define:
 \begin{align}
{\rm{\widetilde{Wres}}}(A)=&\int_X\int_{\bf S}{\mathrm{Tr}}_E\left[p_{-n}(x,\xi)\right]\sigma(\xi){\rm d}x \nonumber\\
&+2\pi\int_ {\partial X}\int_{\bf S'}\left\{{\mathrm{Tr}}_E\left[({\mathrm{Tr}}b_{-n})(x',\xi')\right]+{\mathrm{Tr}}
_F\left[s_{1-n}(x',\xi')\right]\right\}\sigma(\xi'){\rm d}x',
\end{align}
Then~~ a) ${\rm \widetilde{Wres}}([A,B])=0 $, for any
$A,B\in\mathcal{B}$;~~ b) It is a unique continuous trace on
$\mathcal{B}/\mathcal{B}^{-\infty}$.
\end{thm}
 Let $M$ be a compact oriented Riemannian manifold of even dimension $n=2m$. Let $p_{1},p_{2}$ be nonnegative integers and $p_{1}+p_{2}\leq n$.
 Denote by $\sigma_{l}(A)$ the $l$-order symbol of an operator A. An application of (2.1.4) in \cite{Wa1} shows that
\begin{equation}
\widetilde{{\rm Wres}}[\pi^+\widetilde{D}_{F}^{-p_1}\circ\pi^+\widetilde{D}_{F}^{-p_2}]
=\int_M\int_{|\xi|=1}{\rm Tr}_{S(TM)}[\sigma_{-n}(\widetilde{D}_{F}^{-p_1-p_2})]\sigma(\xi)\texttt{d}x+\int_{\partial
M}\Psi,
\end{equation}
where
 \begin{eqnarray}
\Psi&=&\int_{|\xi'|=1}\int^{+\infty}_{-\infty}\sum^{\infty}_{j, k=0}
\sum\frac{(-i)^{|\alpha|+j+k+1}}{\alpha!(j+k+1)!}
{\rm Tr}_{S(TM)\otimes F}
\Big[\partial^j_{x_n}\partial^\alpha_{\xi'}\partial^k_{\xi_n}
\sigma^+_{r}(\widetilde{D}_{F}^{-p_1})(x',0,\xi',\xi_n)\nonumber\\
&&\times\partial^\alpha_{x'}\partial^{j+1}_{\xi_n}\partial^k_{x_n}\sigma_{l}
(\widetilde{D}_{F}^{-p_2})(x',0,\xi',\xi_n)\Big]\rm{d}\xi_n\sigma(\xi')\rm{d}x',
\end{eqnarray}
and $r-k+|\alpha|+\ell-j-1=-n,r\leq-p_{1},\ell\leq-p_{2}$.

When $n=2m$, $ r-k-|\alpha|+l-j-1=-2m,~~r\leq-1,l\leq-2m+2$, and $r=l=-1,~k=|\alpha|=j=0,$ then
\begin{align}
\Psi=&\int_{|\xi'|=1}\int^{+\infty}_{-\infty}
{\rm Tr}_{S(TM)\otimes F}
[ \sigma^+_{-1}(c(u)c(v)c(w)\widetilde{D}_{F}^{-1})(x',0,\xi',\xi_n)\nonumber\\
&\times
\partial_{\xi_n}\sigma_{-2m+2}(\widetilde{D}_{F}^{-2m+2})(x',0,\xi',\xi_n)]\rm{d}\xi_{2m}\sigma(\xi')\rm{d}x'.
\end{align}
An easy calculation gives
  \begin{align}
 \pi^+_{\xi_n}\big( \sigma_{-1}(c(u)c(v)c(w)\widetilde{D}_{F}^{-1})\big)=&
    c(u)c(v)c(w)\pi^+_{\xi_n}\big( \sigma_{-1} (\widetilde{D}_{F}^{-1})\big)\nonumber\\
=&c(u)c(v)c(w) \frac{c(\xi')+ic(\rm{d}x_{n})}{2(\xi_n-i)} \nonumber\\
=& \frac{c(u)c(v)c(w)c(\xi')}{2(\xi_n-i)}+ \frac{ i c(u)c(v)c(w)c(\rm{d}x_{n})}{2(\xi_n-i)},
 \end{align}
where some basic facts and formulae about Boutet de Monvel's calculus which can be found  in Sec.2 in \cite{Wa1}.
By (3.6), we get
  \begin{align}
\partial_{\xi_n}\big(\sigma_{-2m+2}(\widetilde{D}_{F}^{-2m+2})\big)   (x_0)|_{|\xi'|=1}
=\partial_{\xi_n}\big(|\xi|^{2}\big)^{1-m}(x_0)=2(1-m)\xi_n(1+\xi_n^2)^{-m}.
\end{align}
By the relation of the Clifford action and $ {\rm{Tr}}(AB)= {\rm{Tr}}(BA) $, we get
  \begin{align}
{\rm Tr}_{S(TM)}[c(u)c(v)c(w)c(\rm{d}x_{n})]=&\sum_{i,j,k=1}^{n}  u_{j}v_{k}w_{l}
{\rm Tr}_{S(TM)}[c(e_{j})c(e_{k}) c(e_{l})c(\rm{d}x_{n})]\nonumber\\
=&\sum_{i,j,k=1}^{n}  u_{j}v_{k}w_{l} (-\delta_{j}^{l}\delta_{k}^{n}+\delta_{j}^{n}\delta_{k}^{l}+\delta_{j}^{k}\delta_{n}^{l})
{\rm Tr}_{S(TM)}[ {\rm Id}]\nonumber\\
=&\sum_{i,j,k=1}^{n}  (-u_{j}v_{n}w_{j}+u_{n}v_{k}w_{k}+w_{n}u_{j}v_{j})
{\rm Tr}_{S(TM)}[ {\rm Id}]\nonumber\\
=&(u_{n}g(v,w)-v_{n}g(u,w)+w_{n}g(u,v)){\rm Tr}_{S(TM)}[ {\rm Id}].
\end{align}
In the similar way we obtain
\begin{align}
& {\rm Tr}_{S(TM)\otimes F}
[ \sigma^+_{-1}(c(u)c(v)c(w)\widetilde{D}_{F}^{-1})(x',0,\xi',\xi_n) \times
\partial_{\xi_n}\sigma_{-2m+2}(\widetilde{D}_{F}^{-2m+2})(x',0,\xi',\xi_n)] \nonumber\\
=& \frac{(1-m)\xi_n}{2(\xi_n-i)(1+\xi_n^2)^{m}}{\rm Tr}_{S(TM)\otimes F}[c(u)c(v)c(w)c(\xi')]\nonumber\\
&+ \frac{ i(1-m)\xi_n}{2(\xi_n-i)(1+\xi_n^2)^{m}}{\rm Tr}_{S(TM)\otimes F}[c(u)c(v)c(w)c(\rm{d}x_{n})]\nonumber\\
=& \frac{(1-m)\xi_n}{2(\xi_n-i)(1+\xi_n^2)^{m}}\sum_{i=1}^{n-1}\xi_i{\rm Tr}_{S(TM)\otimes F}[c(u)c(v)c(w)c(e_{i})]\nonumber\\
&+ \frac{ i(1-m)\xi_n}{2(\xi_n-i)(1+\xi_n^2)^{m}}(u_{n}g(v,w)-v_{n}g(u,w)+w_{n}g(u,v)){\rm Tr}_{S(TM)\otimes F}[ {\rm Id}].
\end{align}
Also, straightforward computations yield
\begin{align}
\Big(\frac{1}{(\xi_n+i)^{1-m}}\Big)^{(m)}=&-(m-1)\big((\xi_n+i)^{-m}\big)^{(m-1)}
=(-1)^{2}(m-1)m\big((\xi_n+i)^{-m-1}\big)^{(m-2)}\nonumber\\
=& \cdots =(-1)^{m}\frac{(2m-2)!}{(m-2)!}(2i)^{-2m+1},
\end{align}
and
\begin{align}
\Big(\frac{i}{(\xi_n+i)^{m}}\Big)^{(m)}=&-im\big((\xi_n+i)^{-m-1}\big)^{(m-1)}
=  i(-1)^{2}m(m-1)\big((\xi_n+i)^{-m-2}\big)^{(m-2)}\nonumber\\
=& \cdots
=(-1)^{m}\frac{(2m-1)!}{(m-1)!}(2i)^{-2m}.
\end{align}
Combining (4.60) and (4.61), we have
\begin{align}
\Big(\frac{\xi_n}{(\xi_n+i)^{m}}\Big)^{(m)}|_{\xi_n=i}
=&\Big[\Big(\frac{1}{(\xi_n+i)^{m}}\Big)^{(m)}-\Big(\frac{i}{(\xi_n+i)^{m}}\Big)^{(m)}\Big]|_{\xi_n=i}
=\frac{(2m-2)!(-i)2^{-2m}}{(m-1)!}.
\end{align}
Substituting into above calculation we get
\begin{align}
\Psi=&\int_{|\xi'|=1}\int^{+\infty}_{-\infty}
{\rm Tr}_{S(TM)\otimes F}
[ \sigma^+_{-1}(c(u)c(v)c(w)\widetilde{D}_{F}^{-1})(x',0,\xi',\xi_n)\nonumber\\
&\times
\partial_{\xi_n}\sigma_{-2m+2}(\widetilde{D}_{F}^{-2m+2})(x',0,\xi',\xi_n)]\rm{d}\xi_{2m}\sigma(\xi')\rm{ d}x'\nonumber\\
=&\int_{|\xi'|=1}\int^{+\infty}_{-\infty}\Big(
 \frac{(1-m)\xi_n}{2(\xi_n-i)(1+\xi_n^2)^{m}}\sum_{i=1}^{n-1}\xi_i{\rm Tr}_{S(TM)\otimes F}[c(u)c(v)c(w)c(e_{i})]\nonumber\\
&+ \frac{ i(1-m)\xi_n}{2(\xi_n-i)(1+\xi_n^2)^{m}}(u_{n}g(v,w)-v_{n}g(u,w)+w_{n}g(u,v)){\rm Tr}_{S(TM)\otimes F}[ {\rm Id}] \Big)
 \rm{d}\xi_{2m}\sigma(\xi')\rm{d}x'\nonumber\\
=& i(1-m)(u_{n}g(v,w)-v_{n}g(u,w)){\rm Tr}_{S(TM)\otimes F}[ {\rm Id}] {\rm{vol}}(S^{n-2})
 \int^{+\infty}_{-\infty}\frac{\xi_n}{2(\xi_n-i)(1+\xi_n^2)^{m}} \rm{d}\xi_{2m}\rm{d}x'\nonumber\\
 =& \frac{(2m-2)!(1-m)i2^{-2m+1}\pi}{m!(m-1)!}(u_{n}g(v,w)-v_{n}g(u,w)+w_{n}g(u,v)){\rm Tr}_{S(TM)\otimes F}[ {\rm Id}] {\rm{vol}}(S^{n-2})
{\rm d}\rm{vol}_{\partial_{M}}.
\end{align}
Summing up (4.24) and (4.63) leads to the spectral torsion for manifold with boundary  as follows.
\begin{thm}
For the Connes type operator with the trilinear Clifford multiplication by functional of differential one-forms
$c(u),c(v),c(w)$, the spectral torsion for  type-I operator  $D_{F}+(c(T)+\sqrt{-1}c(Y)) \otimes \Phi $ on manifold with boundary  equals to
\begin{align}
\mathscr{T}(c(u),c(v),c(w))
=&\int_M -2^{m+1}T(u,v,w)
{\rm{ Tr}}^{F}( \Phi){\rm{vol}}(S^{n-1}) {\rm d}\rm{vol}_{M}
+\int_{\partial_{M}}\frac{(2m-2)!(1-m)i2^{-2m+1}\pi}{m!(m-1)!}\nonumber\\
&\times(u_{n}g(v,w)-v_{n}g(u,w)+w_{n}g(u,v)){\rm Tr}_{S(TM)\otimes F}[ {\rm Id}] {\rm{vol}}(S^{n-2})
{\rm d}\rm{vol}_{\partial_{M}}.
\end{align}
\end{thm}
\begin{rem}
Let $(M, g)$ be a compact Riemannian manifold, for Dirac operator
 with a fully antisymetric torsion $D_{T}=D-\frac{i}{8}T_{jkl}\gamma^{j}\gamma^{k}\gamma^{l}$,
the torsion functional defined by Dabrowski et al.\cite{DSZ}
recovers the torsion of the linear connection for a canonical spectral triple over a closed spin manifold.
Based on this idea, Theorem 4.13 gives the spectral torsion for manifold with boundary.
\end{rem}

\section*{ AUTHOR DECLARATIONS}
Conflict of Interest£º

The authors have no conflicts to disclose.
\section*{ Author Contributions}
Jian Wang: Investigation (equal); Writing- original draft (equal). Yong Wang: Investigation (equal); Writing-original draft (equal).
\section*{DATA AVAILABILITY}
Data sharing is not applicable to this article as no new data were created or analyzed in this study.

\section*{ Acknowledgements}
The first author was supported by NSFC. 11501414. The second author was supported by NSFC. 11771070.
 The authors also thank the referee for his (or her) careful reading and helpful comments.


\begin{thebibliography}{00}
\bibitem{DSZ} L. Dabrowski, A. Sitarz, and P. Zalecki.: Spectral Torsion, Comm. Math. Phys. 405, 130, (2024).

\bibitem{AT} T. Ackermann and J. Tolksdorf.:  A generalized Lichnerowicz formula, the Wodzicki residue and gravity.
J. Geom. Physics. 19, 143-150, (1996).
\bibitem{PS1} F. Pf$\ddot{a}$ffle and C. A. Stephan.: Chiral Asymmetry and the Spectral Action.
  Commun. Math. Phys. 321, 283-310, (2013).
\bibitem{WW2} J. Wang, Y. Wang, C. Yang.: Dirac operators with torsion and the noncommutative
residue for manifolds with boundary. J. Geom. Physics. 81, 92-111, (2014).

\bibitem{MA} M. Adler.: On a trace functional for formal pseudo-differential operators and the symplectic
structure of Korteweg-de Vries type equations, Invent. Math. 50, 219-248,(1979).
\bibitem{Wo} M. Wodzicki.: local invariants of spectral asymmetry. Invent. Math. 75(1), 143-178, (1984).
\bibitem{Wo1} M. Wodzicki.:  Non-commutative residue I, Lecture Notes in Math., Springer, New York, Vol. 1289,
320-399, (1987).
 \bibitem{Co1} A. Connes.: The action functinal in Noncommutative geometry. Comm. Math. Phys. 117, 673-683, (1988).
 \bibitem{Co2} A. Connes.: Quantized calculus and applications.  XIth International Congress of Mathematical Physics(Paris,1994),
 Internat Press, Cambridge, MA, 15-36, (1995).
\bibitem{KW} W. Kalau and M. Walze.: Gravity, Noncommutative geometry and the Wodzicki residue. J. Geom. Physics. 16, 327-344,(1995).
\bibitem{Ka} D. Kastler.: The Dirac Operator and Gravitation. Comm. Math. Phys. 166, 633-643, (1995).

\bibitem{Ac} T. Ackermann.: A note on the Wodzicki residue. J. Geom. Phys. 20, 404-406, (1996).
\bibitem{FGLS} B. V. Fedosov, F. Golse, E. Leichtnam, E. Schrohe.: The noncommutative residue for manifolds with boundary. J. Funct. Anal.
142, 1-31, (1996).
 \bibitem{S} E. Schrohe.:  Noncommutative residue, Dixmier's trace, and heat trace expansions on manifolds with boundary. Contemp. Math.
242, 161-186, (1999).
\bibitem{Wa1} Y. Wang.: Diffential forms and the Wodzicki residue for Manifolds with Boundary. J. Geom. Physics. 56, 731-753, (2006).
\bibitem{Wa2} Y. Wang.: Gravity and the Noncommutative Residue for Manifolds with Boundary. Lett. Math. Phys. 80, 37-56, (2007).
\bibitem{Wa3} Y. Wang.: Lower-Dimensional Volumes and Kastler-kalau-Walze Type Theorem for Manifolds with Boundary.
      Commun. Theor. Phys. 54, 38-42, (2010).
\bibitem{SZ} A. Sitarz and A. Zajac.: Spectral action for scalar perturbations
of Dirac operators, Lett. Math. Phys., 98(3), 333-348, (2011).
\bibitem{IL} B. Iochum and C. Levy.: Tadpoles and commutative spectral triples, J. Noncommut. Geom., 5(3), 299-329, (2011).
\bibitem{Wa4} Y. Wang.: A Kastler-Kalau-Walze type theorem and the Spectral action for perturbations of Dirac operators on manifolds
with boundary. Abstr. Appl. Anal. 2014, 619120,1-13, (2014).
\bibitem{WW2} J. Wang  and Y. Wang.: Twisted Dirac operators and the noncommutative residue for manifolds with boundary[J],
J.Pseudo-Differ.Appl., 7(2):181-211, (2016).
\bibitem{WWW} T. Wu, J. Wang, Y. Wang.: Dirac-Witten Operators and the Kastler-Kalau-Walze Type
Theorem for Manifolds with Boundary. J. Nonlinear Math. Phys. 29, 1-40, (2022).
\bibitem{MNZ}  P. Martinetti,  G. Nieuviarts,  R. Zeitoun.: Torsion and Lorentz symmetry from twisted spectral triples. arXiv:2401.07848v5. (2024).
\bibitem{Y} Y. Yu.: The Index Theorem and The Heat Equation Method, Nankai Tracts in Mathematics, Vol.2, World Scientific Publishing, (2001).
\bibitem{GHV} W. Greub, S. Halperin and R.Vanstone.: Connections, Curvature and Cohomology, Vo.1 (Academic press, New York) (1976).
\bibitem{WW1}  J. Wang and Y. Wang.: A general Kastler-Kalau-Walze type theorem for manifolds with boundary,
Int. J. Geom. Methods Mod. Phys. 13, no. 1, 1650003, 16 pp, (2016).
\bibitem{YW1} Y.C. Yang  and Y. Wang.: The general Dabrowski-Sitarz-Zalecki type theorem for
odd dimensionalmanifolds with boundary III. J. Pseudo-Differ. Oper. Appl. 15:41, (2024).



\end{thebibliography}
\end{document}